\documentclass[12pt, draftclsnofoot, onecolumn]{IEEEtran}

\usepackage{graphics} 
\usepackage{amssymb}  
\usepackage{cite}
\usepackage{amsmath}
\usepackage{algorithm}
\usepackage{algorithmic}
\usepackage{graphicx}
\usepackage{epstopdf}
\usepackage{booktabs}
\usepackage{multirow}
\usepackage{etoolbox}
\usepackage{mathdots}
\usepackage{xcolor}
\makeatletter
\patchcmd{\@makecaption}
{\scshape}
{}
{}
{}
\makeatother
\usepackage{diagbox}
\usepackage{bm}
\usepackage{subfigure}
\usepackage{lipsum}
\usepackage{braket,amsfonts,amsopn} 
\usepackage{mathrsfs,enumerate}
\usepackage{url}
\usepackage{hyperref}

\def\norm #1{\left\|#1\right\|}

\def\frobn #1{\left\|#1\right\|_{\text{F}}}

\def\abs #1{\left|#1\right|}

\def\st{\text{subject to }}

\def\bC{\mathbb{C}}

\def\bR{\mathbb{R}}

\def\bS{\mathbb{S}}

\def\m #1{\boldsymbol{#1}}

\def\cA{\mathcal{A}}

\def\cH{\mathcal{H}}

\def\cL{\mathcal{L}}
\def\cM{\mathcal{M}}

\def\cP{\mathcal{P}}

\def\cS{\mathcal{S}}
\def\cT{\mathcal{T}}

\def\bee{\begin{equation}}
	\def\ene{\end{equation}}

\def\beq{\begin{eqnarray}}
	\def\enq{\end{eqnarray}}
\def\lentwo{\setlength\arraycolsep{2pt}}

\newtheorem{thm}{Theorem}
\newtheorem{prop}{Proposition}

\newenvironment{proof}{{\noindent\it \quad Proof:}\;}{\hfill $\square$\par}

\def\equ #1{\begin{equation}#1\end{equation}}
\def\equa #1{\begin{eqnarray}#1\end{eqnarray}}
\def\sbra #1{\left(#1\right)}
\def\mbra #1{\left[#1\right]}
\def\lbra #1{\left\{#1\right\}}
\def\diag #1{\text{diag}#1}
\def\tr #1{\text{tr}#1}

\def\rank #1{\text{rank}#1}
\def\st {\text{ subject to }}

\DeclareMathOperator*{\argmin}{arg\,min}

\title{Direction-of-Arrival Estimation for Constant Modulus Signals Using a Structured Matrix Recovery Technique}
\author{Xunmeng Wu, Zai Yang, Zhiqiang Wei, and Zongben Xu
	\thanks{
		The research of the project was supported by the National Natural Science Foundation of China under Grants 61977053 and 11922116.
		
		The authors are with the School of Mathematics and Statistics, Xi'an Jiaotong University, Xi'an 710049, China (e-mails: wxm1996@stu.xjtu.edu.cn, yangzai@xjtu.edu.cn, zhiqiang.wei@xjtu.edu.cn, zbxu@xjtu.edu.cn).
		
		Author for correspondence: Zai Yang.
	}
}

\begin{document}\maketitle
	
	
	\begin{abstract} 
		This paper addresses the problem of direction-of-arrival (DOA) estimation for constant modulus (CM) source signals using a uniform or sparse linear array. Existing methods typically exploit either the Vandermonde structure of the steering matrix or the CM structure of source signals only. In this paper, we propose a \emph{s}tructured \emph{ma}trix \emph{r}ecovery \emph{t}echnique (SMART) for CM DOA estimation via fully exploiting the two structures. In particular, we reformulate the highly nonconvex CM DOA estimation problems in the noiseless and noisy cases as equivalent rank-constrained Hankel-Toeplitz matrix recovery problems, in which the Vandermonde structure is captured by a series of Hankel-Toeplitz block matrices, of which the number equals the number of snapshots, and the CM structure is guaranteed by letting the block matrices share a same Toeplitz submatrix. The alternating direction method of multipliers (ADMM) is applied to solve the resulting rank-constrained problems and the DOAs are uniquely retrieved from the numerical solution. Extensive simulations are carried out to corroborate our analysis and confirm that the proposed SMART outperforms state-of-the-art algorithms in terms of the maximum number of locatable sources and statistical efficiency.
	\end{abstract}
	
	
	\section{Introduction}
	Direction-of-arrival (DOA) estimation, which refers to the process of estimating the DOAs of multiple farfield narrowband source signals from observed snapshots of a sensor array, is a fundamental problem in array signal processing and has numerous applications in sensing systems such as radar \cite{zheng2017adaptive} and sonar \cite{knight1981digital}, wireless communications \cite{guo2017millimeter,tsai2018millimeter,shafin2016doa,wei2018multi,wei2019noma}, and integrated sensing and communications (ISAC) \cite{zhang2021overview,liu2022survey}. 
	For many communication signals, such as those based on phase modulation (PM) and frequency modulation (FM) in the analog domain, and phase shift keying (PSK) and frequency shift keying (FSK) for digital signals, constant modulus (CM) is an appealing property that enables a high amplification efficiency and a large coverage area \cite{mancuso2011reducing,liu2018toward}.
	CM signals have played crucial roles in many tasks for communications such as blind equalization \cite{godard1980self,treichler1983new,mariere2003blind}, blind beamforming \cite{treichler1985new,van1996analytical}, intelligent reflect surface design \cite{wu2019intelligent,wei2020sum,he2021novel}, waveform design in ISAC \cite{liu2018toward,hassanien2019dual}, etc. For instance, considering an ISAC system in which communications and radar sensing are simultaneously carried out, CM transmitted signals can not only lead to an energy-efficient transmission but also enhance the sensing performance \cite{zhang2021overview,liu2022survey}. This motivates us to be concerned about the DOA estimation problem for CM source signals using a uniform or sparse linear array (ULA or SLA) in this paper, which is referred to as the CM DOA estimation problem. 
	
	Exploiting the CM structure of source signals brings great advantages to DOA estimation in terms of the maximum number of locatable sources and the Cram\'{e}r-Rao bound (CRB) \cite{wax1992unique,williams1992resolving,valaee1994alternative,leshem1999direction}. In particular, necessary and sufficient conditions for unique localization of CM signal sources using a ULA have been analyzed in the absence of noise in \cite{wax1992unique,williams1992resolving,valaee1994alternative}. It was shown that the maximum number of CM sources that can be uniquely localized can exceed the number of sensors. This result breaks through the conventional limit where the number of sources must be less than the number of sensors for arbitrary source signals \cite{pisarenko1973retrieval,bresler1986number}. Moreover, it was shown in \cite{leshem1999direction} that the CRB of DOA estimation for CM signals is significantly smaller than that for arbitrary signals in challenging scenarios with limited snapshots or closely located sources. Hence, it is of great interest to develop algorithms via utilizing the CM structure to improve the DOA estimation performance. 
	
	The CM DOA estimation problem is highly nonconvex and NP-hard, in which the high nonlinearity and the nonconvexity arise both from the Vandermonde structure of the steering matrix for ULAs and the CM structure of source signals. Taking account of the CM structure into DOA estimation dates back to \cite{gooch1986cm,shynk1996constant} and has been studied in the past four decades. 
	Two-step CM DOA estimation methods have been proposed in \cite{leshem1999direction,shynk1996constant,van2001asymptotic}, where the iterative CM algorithm (CMA) \cite{godard1980self,treichler1983new,treichler1985new}, the analytic CM algorithm (ACMA) \cite{van1996analytical}, or the zero-forcing variant of ACMA (ZF-ACMA) \cite{van2001asymptotic} is used to estimate the steering matrix by exploiting the CM structure, from which the DOAs are estimated.
	All of these two-step methods are suboptimal because the Vandermonde structure is not used to solve for the steering matrix in the first step. Subsequently, Leshem \cite{leshem2000maximum} has derived a Newton scoring algorithm for maximum likelihood estimation (MLE) that work directly in the parameter domain. This local optimization algorithm approximately solves a nonlinear least-squares (NLS) problem of the MLE and its performance heavily depends on the initialization. Besides, Stoica et al. \cite{stoica2000maximum} have derived a simple expression for the MLE in the single source case. But to the best of our knowledge, few substantial improvements have been further made in the past two decades. In this paper, we aim at fully using both the Vandermonde and CM structures to substantially improve the performance of CM DOA estimation.
	
	Note that the CM structure makes sense only in the presence of multiple snapshots. In the single-snapshot case, the CM DOA estimation problem concerned in this paper is nothing but the general line spectral estimation problem \cite{stoica2005spectral,tang2013compressed,wu2022maximum} in which the exploitation of the Vandermonde structure has been extensively studied. This motivates us to extend existing algorithms for line spectral estimation from the single-snapshot to the multiple-snapshot case by preserving the CM structure. However, such an extension is challenging since the CM structure of source signals in our problem results in a large number of nonconvex constraints, as mentioned previously, while previous studies \cite{ottersten1993exact,krim1996two,stoica2005spectral,yang2018sparse,starer1992newton,feder1988parameter,stoica1990maximum,schmidt1986multiple,barabell1983improving,roy1989esprit,malioutov2005sparse,hu2013doa,stoica2010spice,yang2014discretization,wu2017toeplitz,yang2018sample,yang2022robust} on the multiple-snapshot DOA estimation have not considered the CM structure.
	
	This paper is inspired by our recent work \cite{wu2022maximum} on line spectral estimation (or single-snapshot DOA estimation). In \cite{wu2022maximum}, a rank-constrained Hankel-Toeplitz block matrix recovery approach was proposed where the matrix rank equals the number of sources and the DOAs and amplitudes (or moduli) of all sources are fully encoded in the Toeplitz submatrix. To extend \cite{wu2022maximum} to the multiple-snapshot CM DOA estimation concerned in the present paper, we propose to form a Hankel-Toeplitz block matrix with respect to each snapshot and let all the block matrices share a same Toeplitz submatrix so that the same DOAs and associated moduli are imposed among all snapshots. Consequently, the CM DOA estimation problem is formulated as a rank-constrained structured matrix recovery problem either in the noiseless or noisy case in which the only nonconvexity arises due to the rank constraints on the Hankel-Toeplitz block matrices and the Toeplitz submatrix. We show that the rank constraints on the Toeplitz submatrix can be removed with good performance and the resulting rank-constrained problems can be solved efficiently using the alternating direction method of multipliers (ADMM) in which both subproblems are solved in closed form. The proposed approach is suitable for both ULA and SLA cases and is termed as structured matrix recovery technique (SMART). Extensive numerical results are provided to verify our theoretical findings and demonstrate the advantages of the proposed SMART in terms of the number of locatable sources and accuracy.
	
	\subsection{Related Work}
	The studies on the CM structure date back to CMA \cite{godard1980self,treichler1983new,treichler1985new} for blind equalization and beamforming in the 1980s, where a nonconvex modulus-error cost function is minimized using stochastic gradient descent. 
	To circumvent the convergence issue of CMA, ACMA and ZF-ACMA were proposed in \cite{van1996analytical,van2001asymptotic} using standard linear algebra tools for blind source separation. A ZF-ACMA-based two-step method was proposed in \cite{amar2010low} for blindly estimating polynomial phase signals observed by a sensor array of which the structure is unknown. 	A two-dimensional ACMA-based estimator was recently developed in \cite{hamici2019elements} to localize CM sources using two perpendicular Centro-symmetric arrays under array element failures. Both the problems in \cite{amar2010low,hamici2019elements} are different from CM DOA estimation concerned in this paper.
	
	Another line of research is to consider a CM-constrained nonconvex optimization problem in which the CM structure is incorporated into a constraint. In \cite{smith1999optimum,leshem2000maximum,wang2012design}, an explicit argument/phase parameterization of the CM constraint was employed, followed by derivative-based methods, such as gradient descent and Newton's method, to solve the resulting unconstrained nonconvex optimization problem. On the other hand, since the CM constraint is a special case of quadratic constraints, the semidefinite relaxation (SDR) technique \cite{luo2010semidefinite} was applied to give a convex relaxation of the CM-constrained problem by deriving an equivalent formulation and then dropping its nonconvex rank-one constraint. Moreover, nonconvex methods, including power method-like iterations \cite{soltanalian2014designing} and gradient projection algorithm \cite{tranter2017fast}, were introduced for local optimization of the special unit-modulus-constrained problem. Many other related techniques were also proposed to tackle the CM constraint in applications, such as relaxing the CM equality constraint into the inequality one \cite{aldayel2016successive,he2021novel}, introducing auxiliary variables in the ADMM framework \cite{fan2018constant}, and using the branch-and-bound algorithm \cite{liu2018toward}, etc. But these methods only focus on the CM structure and do not take the Vandermonde structure of the steering matrix for linear arrays into consideration.
	
	The Vandermonde structure for linear arrays is a key property in DOA estimation. During the past four decades, a great number of approaches have been developed to utilize this structure in DOA estimation for arbitrary source signals, see \cite{ottersten1993exact,krim1996two,stoica2005spectral,yang2018sparse}. By modeling the source signals as unknown deterministic quantities or Gaussian random processes, deterministic or stochastic MLE methods were developed in which the main difficulty comes from the nonlinearity and nonconvexity with respect to the DOAs. Several methods were applied to approximately solve MLE in the DOA parameter domain, such as Newton's method \cite{starer1992newton}, expectation maximization \cite{feder1988parameter}, and the method of direction estimation (MODE) \cite{stoica1990maximum}, etc. To circumvent the difficulty of MLE, subspace-based methods were proposed based on the eigenstructure of the array covariance matrix, including MUSIC \cite{schmidt1986multiple}, Root-MUSIC \cite{barabell1983improving}, and ESPRIT \cite{roy1989esprit,yang2022nonasymptotic} which exploit the Vandermonde structure by finding peaks of the spectrum, polynomial rooting, and using the so-called shift structure, respectively. 
	
	Sparse optimization and compressed sensing approaches have dominated the research of this century on DOA estimation, see \cite{yang2018sparse}. To overcome the aforementioned nonlinearity, the continuous DOA domain is approximated by a given set of grid points and the steering matrix turns to be a fixed Vandermonde dictionary matrix. Based on this idea, sparse signal recovery techniques such as $\ell_1$-norm minimization \cite{malioutov2005sparse,hu2013doa} and the semiparametric iterative covariance-based estimation (SPICE) method \cite{stoica2010spice} were proposed. Another advanced technique to tackle the nonlinearity and to exploit the Vandermonde structure is based on the Carath\'{e}odory-Fej\'{e}r theorem \cite[Theorem 11.5]{yang2018sparse} that says that a low-rank positive-semidefinite (PSD) Toeplitz matrix admits a unique Vandermonde decomposition in which the DOAs are given by the nodes of the Vandermonde matrix. Several representative methods include sparse and parametric approach (SPA) \cite{yang2014discretization}, covariance matrix reconstruction approach \cite{wu2017toeplitz}, and atomic norm minimization (ANM) \cite{tang2013compressed,yang2018sample}, which explicitly optimize
	a low-rank PSD Toeplitz matrix where the low-rankness comes from the fact that sources are usually less than sensors. But all these mainstream methods are designed for arbitrary source signals and show suboptimal performance when applied to DOA estimation for CM source signals.
	
	\subsection{Notation}
	Boldface letters are reserved for vectors and matrices. The sets of real and complex numbers are denoted by $\bR$ and $\bC$, respectively. For vector $\m{x}$, $\m{x}^T$, $\overline{\m{x}}$, $\m{x}^H$, and $\norm{\m{x}}_2$ denote its transpose, complex conjugate, conjugate transpose, and $\ell_2$ norm, respectively. For matrix $\m{X}$, the transpose, complex conjugate, conjugate transpose, inverse, Moore-Penrose pseudo-inverse, Frobenius norm, rank, column space and trace are denoted by $\m{X}^T$, $\overline{\m{X}}$, $\m{X}^H$, $\m{X}^{-1}$, $\m{X}^{\dag}$, $\frobn{\m{X}}$, $\rank\sbra{\m{X}}$, $\text{range}\sbra{\m{X}}$ and $\tr\sbra{\m{X}}$ respectively. The inner product of matrices $\m{X}$ and $\m{Y}$ is defined as $\langle \m{X},\m{Y} \rangle_{\bR} = \Re\lbra{\tr\sbra{\m{X}^H\m{Y}}}$. $\m{X}\geq \m{0}$ means that $\m{X}$ is Hermitian PSD. The notation $\abs{\cdot}$ denotes the modulus of a scalar or the cardinality of a set, and the diagonal matrix with vector $\m{x}$ on the diagonal is denoted by $\diag\sbra{\m{x}}$. The $j$th entry of vector $\m{x}$ is $x_j$, the $\sbra{i,j}$th entry of $\m{X}$ is $X_{i,j}$, and the $i$-th row (or $j$-th column) of $\m{X}$ is $\m{X}_{i,:}$ (or $\m{X}_{:,j}$). An identity matrix is denoted by $\m{I}$. 
	Denote $\cH$ as a Hankel operator that maps $\m{x} \in \bC^{N}$ to an $n \times n$ Hankel matrix $\cH\m{x}$ with its $\sbra{i,j}$ entry given by $x_{i+j-1}$ where $N=2n-1$. Denote $\cT$ as a Hermitian Toeplitz operator that maps $\m{t}\in\bC^n$ to an $n\times n$ Hermitian Toeplitz matrix $\cT\m{t}$ with its $\sbra{i,j},i\ge j$ entry given by $t_{i-j+1}$.
	
	\subsection{Organization} \label{sec:org}
	The rest of the paper is organized as follows. Section \ref{sec:data} presents the data model and prior art. Section \ref{sec:main} presents the proposed model of SMART. Section \ref{sec:alg} presents the ADMM algorithm of SMART. Section \ref{sec:extensions} extends SMART to the large number of sources case and the SLA case. Section \ref{sec:sim} provides numerical simulations and Section \ref{sec:conclusion} concludes this paper.
	
	\section{Data Model and Prior Art} \label{sec:data}
	\subsection{Data Model}
	Assume that $K$ narrowband farfield CM source signals impinge on a ULA of $N$ sensors from DOAs $\theta_k \in \left[-\frac{\pi}{2}, \frac{\pi}{2} \right), k=1,\ldots,K$. Let $L$ be the number of snapshots. At the $l$-th time instant, the array observation can be modeled as \cite{leshem1999direction}
	\equ{
		\begin{split}
			\m{y}\sbra{l} & = \sum^K_{k=1} \m{a}\sbra{\theta_k} b_k e^{i\phi_{k,l}} +\m{e}\sbra{l}, \quad l = 1,\ldots,L,
		\end{split}
		\label{eq:signalmodel}
	} 
	where $\m{y}\sbra{l}\in\mathbb{C}^N$, $\m{e}\sbra{l}\in\mathbb{C}^N$ is the noise vector, $b_k e^{i\phi_{k,l}}$ is the $k$-th CM source signal at time $l$ where $i=\sqrt{-1}$, $b_k>0$ is the unknown modulus and $\phi_{k,l}$ is the unknown phase. It is seen that for any $k$-th source, all snapshots share the same parameters $\theta_k$ and $b_k$. The steering vector of the $k$-th source is given by
	\equ{
		\m{a}\sbra{\theta_k} = \mbra{1,e^{i2\pi\sin\sbra{\theta_k}d/\lambda},\ldots,e^{i2\pi\sbra{N-1}\sin\sbra{\theta_k}d/\lambda} }^T,
	}
	where the intersensor spacing $d$ is taken as half the wavelength $\lambda$. 
	
	Denote the DOAs $\m{\theta}=\mbra{\theta_1,\ldots,\theta_K}^T$. The steering vectors compose the steering matrix $\m{A}(\m{\theta})=\mbra{\m{a}\sbra{\theta_1},\ldots,\m{a}\sbra{\theta_K}}$ that has the Vandermonde structure. We rewrite \eqref{eq:signalmodel} in matrix form as:
	\equ{
		\m{Y} = \m{A}(\m{\theta})\m{B}\m{\Phi} + \m{E} = \m{X}^\star + \m{E}, \label{eq:signal_matrix_CM}
	}
	where $\m{Y}=\mbra{\m{y}\sbra{1},\ldots,\m{y}\sbra{L}}$, $\m{E}=\mbra{\m{e}\sbra{1},\ldots,\m{e}\sbra{L}}$, $\m{B}=\diag\sbra{\m{b}}$ with $\m{b}=\mbra{b_1,\ldots,b_K}^T$, $\m{\Phi}\in\bC^{K\times L}$ with $\Phi_{k,l} = e^{i\phi_{k,l}}$ and $\m{X}^\star = \m{A}(\m{\theta})\m{B}\m{\Phi}$. 
	
	The data model for the case of a SLA is given by
	\equ{
		\m{Y}_{\Omega} = \cP_{\Omega}(\m{Y}) = \m{A}_{\Omega}\sbra{\m{\theta}} \m{B} \m{\Phi}+\m{E}_{\Omega} = \m{X}_{\Omega}^\star + \m{E}_{\Omega} , \label{eq:SLA_data}
	}
	where the subset $\Omega \subseteq \lbra{1,\ldots,N}$ of size $M \le N$ is the index set of a SLA, $\cP_{\Omega}$ is the projection onto the rows supported on $\Omega$ that sets all entries outside of $\Omega$ to zero, $\m{A}_{\Omega}(\m{\theta}) = \cP_{\Omega}(\m{A}(\m{\theta}))$ and $\m{E}_{\Omega} = \cP_{\Omega}(\m{E})$. Evidently, \eqref{eq:SLA_data} degenerates to \eqref{eq:signal_matrix_CM} for the ULA when $M=N$. 
	
	Given the array output $\m{Y}$ in \eqref{eq:signal_matrix_CM} (or $\m{Y}_{\Omega}$ in \eqref{eq:SLA_data}), we focus on estimating the DOAs $\m{\theta}$ by exploiting the CM structure of the source signals $\m{B}\m{\Phi}$ and the Vandermonde structure of the steering matrix $\m{A}(\m{\theta})$ in this paper. 
	The source number $K$ is assumed to be known and its determination, if it is unknown, is another fundamental problem \cite{stoica2004model}.
	
	\subsection{Existing Methods for CM DOA Estimation}
	Before introducing the proposed approach, we briefly discuss two existing CM DOA estimation methods \cite{leshem1999direction,leshem2000maximum,stoica2000maximum}:
	
	\subsubsection{ACMA-based method}
	For CM DOA estimation, the ACMA approach \cite{van1996analytical} has been applied to estimate the steering matrix $\m{A}(\m{\theta})$ in \cite{leshem1999direction}. In ACMA, the CM factorization problem was considered in the noiseless case:
	\equ{
		\text{find} \; \m{W}\in \bC^{K\times N}, \st \m{W}\m{X}^\star = \m{S}, \; \abs{S_{k,l}} = 1, \label{eq:ACMA1}
	}
	and the problem of minimizing the modulus-error cost function was considered in the noisy case:
	\equ{
		\min_{\m{w}\in \bC^{1\times N}} \sum^L_{l=1} \sbra{ \abs{\sbra{\m{w}\m{Y}}_l}^2-1 }^2, \label{eq:ACMA2}
	}
	where $\m{w}$ is the $k$-th row vector of $\m{W}$ and $\sbra{\cdot}_l$ is the $l$-th element of a vector.
	The problem in \eqref{eq:ACMA1} (or \eqref{eq:ACMA2}) was shown in \cite{van1996analytical} to be a generalized eigenvalue problem (or generalized Schur decomposition) and was solved by a simultaneous diagonalization of a set of matrices. The steering matrix $\m{A}(\m{\theta})$ was estimated as $\m{W}^{\dag}$ from which the DOAs were estimated via one-dimensional search or ESPRIT \cite{leshem1999direction}. Note that this ACMA-based DOA estimation method is suboptimal since the Vandermonde structure of $\m{A}(\m{\theta})$ is not used to solve for $\m{W}$.
	
	\subsubsection{Newton's method} \label{sec:NM}
	To improve the estimation performance, a Newton scoring algorithm has been derived for the MLE of CM DOA estimation in \cite{leshem2000maximum}. In particular, supposing that the DOAs $\m{\theta}$, moduli $\m{b}$ and phases $\m{\Phi}$ are deterministic but unknown parameters and assuming i.i.d. circular complex Gaussian noise, the MLE in the parameter domain is to tackle the NLS problem \cite{leshem2000maximum,stoica2000maximum}:
	\equ{
		\min_{\m{\theta},\m{b},\m{\Phi}} \frobn{\m{Y}-\m{A}(\m{\theta})\m{B}\m{\Phi}}^2. \label{eq:NLS_CM}
	}
	Due to the nonlinearity and the nonconvexity of the cost function in \eqref{eq:NLS_CM}, a Newton local optimization algorithm with respect to the parameters $\lbra{\m{\theta},\m{b},\m{\Phi}}$ has been devised to solve \eqref{eq:NLS_CM} and suboptimal methods including ESPRIT and ACMA were chosen for initialization. But we note that the performance of this algorithm heavily depends on the initialization.
	
	\section{Proposed CM DOA estimation Approach} \label{sec:main}
	In this section, we formulate the CM DOA estimation problems without or with noise using a ULA as equivalent rank-constrained Hankel-Toeplitz matrix recovery problems by making full use of the Vandermonde and CM structures.
	Without loss of generality, we assume that the number of sensors $N$ is odd such that $N=2n-1$ for integer $n$. If $N$ is even, then we may assume that the $\sbra{N+1}$st sensor fails to work, which can be regarded as the SLA case and will be tackled in Section \ref{sec:extensions}. Moreover, we assume that the number of sources is small, i.e., $K<n$ and the case of $K\ge n$ will also be tackled in Section \ref{sec:extensions}.
	
	\subsection{Problem Statement}
	Denote the set of CM spectrally sparse signals as
	\equ{
		\begin{split}
			\cS_0 =
			& \left\{ \m{A} \sbra{\m{\theta}}\m{B}\m{\Phi} \in \bC^{N\times L}: \;  \theta_k \in \left[-\frac{\pi}{2}, \frac{\pi}{2} \right), \m{B}=\diag\sbra{\m{b}}, 
			\right.\\
			& \phantom{\bC} \qquad \qquad b_k \ge 0, \Phi_{k,l} = e^{i\phi_{k,l}}, k=1,\ldots,K,  l = 1,\ldots,L \bigg\}.
		\end{split}  \label{eq:S0}
	}
	Note that the signal is called spectrally sparse since the number of sources $K$ is small.
	In the absence of noise, the CM DOA estimation problem is to 
	find a factorization satisfying the Vandermonde and the CM constraints, which is given by
	\equ{
		\begin{split}
			\text{find} \ \lbra{\m{\theta},\m{b},\m{\Phi} }, \st & \m{A}(\m{\theta})\m{B}\m{\Phi} = \m{X}^\star \in \cS_0. \label{eq:noiseless}
		\end{split}
	} 
	In the presence of noise, we consider the MLE and reformulate the NLS problem in \eqref{eq:NLS_CM} as the following recovery problem in the signal domain
	\equ{\min_{\m{X}} \frobn{\m{Y} - \m{X}}^2, \st \m{X}\in \cS_0. \label{eq:p_S0_CM}}
	
	It is challenging to handle $\m{X}\in \cS_0$ due to the high nonconvexity with respect to $\lbra{\theta_k, \Phi_{k,l}}$ in $\cS_0$. To overcome these difficulties and make the optimization problems in \eqref{eq:noiseless} and \eqref{eq:p_S0_CM} tractable, the key is to characterize the constraint $\m{X}\in \cS_0$ exactly and efficiently.
	
	\subsection{Proposed Structured Matrix Optimization Model} \label{sec:proposed}
	Let us first recall the single-snapshot case. Recall that $N=2n-1$ and we assume that $K<n$ in this section. A rank-constrained Hankel-Toeplitz matrix model has been proposed in our recent work \cite{wu2022maximum} to characterize the set $\cS_0$ with $L=1$, which is given by
	\equ{
		\begin{split}
			\cS^1_{\text{HT}} =
			\lbra{ \m{x}:\,\begin{bmatrix} \cT\overline{\m{t}} & \cH\overline{\m{x}} \\ \cH\m{x} & \cT\m{t} \end{bmatrix}\geq \m{0},\;\rank\begin{bmatrix} \cT\overline{\m{t}} & \cH\overline{\m{x}} \\ \cH\m{x} & \cT\m{t} \end{bmatrix}\leq K  \text{ for some } \m{t}\in\bC^ n }.
		\end{split} \label{eq:S_HT} 
	}
	Moreover, for a given $\m{x}=\sum^K_{k=1} \m{a}(\theta_k) b_k e^{i\phi_k} \in \cS_{\text{HT}}^1$ with nonzero $b_k$'s, there exists a unique $\m{t}$ in \eqref{eq:S_HT} satisfying
	\equ{
		\cT\m{t} = \m{A}_n(\m{\theta}) \m{B} \m{A}_n^H(\m{\theta}), \label{eq:T_Vander}
	} 
	where $\m{A}_{n}(\m{\theta})=\mbra{\m{a}_{n}\sbra{\theta_1},\ldots,\m{a}_{n}\sbra{\theta_K}} \in \bC^{n\times K}$ with $\m{a}_{n}\sbra{\theta_k} = \mbra{1,e^{i\pi\sin \theta_k},\ldots,e^{i\pi\sbra{n-1}\sin \theta_k} }^T$.
	This means that both the DOAs $\m{\theta}$ and moduli $\m{b}$ are fully preserved by the low-rank PSD Toeplitz submatrix $\cT\m{t}$ in the Hankel-Toeplitz block matrix.
	
	We then focus on the multi-snapshot case where all snapshots share the same DOAs and moduli due to the Vandermonde and CM structures. For each snapshot $l=1,\ldots,L$, we construct a low-rank PSD Hankel-Toeplitz matrix as in \eqref{eq:S_HT}. To impose the same DOAs and moduli among the snapshots, the same Toeplitz matrix is used in the Hankel-Toeplitz matrices. In particular, the Hankel-Toeplitz matrix for the $l$-th snapshot is given by
	\equ{
		\cM(\m{X}_{:,l},\m{t}) \triangleq \begin{bmatrix} \cT\overline{\m{t}} & \cH\overline{\m{X}_{:,l}} \\ \cH\m{X}_{:,l} & \cT\m{t} \end{bmatrix}, \label{eq:HT2}
	}
	where $\cH\m{X}_{:,l}$ and $\cT\m{t}$ are $n\times n$ Hankel and Hermitian Toeplitz matrices, respectively, and 
	\equ{
		\begin{split}
			\cS_{\text{HT}} = & \lbra{ {\m{X}:\rank\sbra{\cT\m{t}} = \rank\sbra{\cM(\m{X}_{:,l},\m{t})}\leq K,} \right. \\
				&  \left.  
				{ \qquad  \cM(\m{X}_{:,l},\m{t})\geq \m{0}, l=1,\ldots,L, \text{for some } \m{t} \in \bC^n } }. \label{eq:S_1} 
		\end{split}
	}
	Formally, we have the following theorem, which generalizes \cite[Theorem 1]{wu2022maximum} for the single-snapshot case.
	\begin{thm} \label{thm:K<n}
		If $K<n$, then we have the following conclusions:
		\begin{enumerate}
			\item $\cS_{\text{HT}} = \cS_0$;
			\item For any $\m{X} = \m{A} \sbra{\m{\theta}}\m{B}\m{\Phi} \in \cS_0$ with distinct $\lbra{\theta_k}$ and $\lbra{b_k > 0}$, $\cT\m{t}$ in \eqref{eq:S_1} is unique with
			\equ{
				\cT\m{t} = \m{A}_n(\m{\theta}) \m{B} \m{A}_n^H(\m{\theta}). \label{eq:Tt}
			}
		\end{enumerate}  
	\end{thm}
	\begin{proof}
		We first show $\cS_0 \subset \cS_{\text{HT}}$. For any given $\m{X}=\m{A}(\m{\theta})\m{B}\m{\Phi} \in \cS_0$, it can be easily shown that
		\equ{
			\cH\m{X}_{:,l} = \m{A}_{n}\sbra{\m{\theta}} \cdot \diag(\m{B} \m{\Phi}_{:,l}) \cdot \m{A}_{n}^T\sbra{\m{\theta}}, \; l = 1,\ldots,L. \label{eq:Hankel2}
		} 
		For convenience, $\m{A}_{n}(\m{\theta})$ is rewritten as $\m{A}$ hereafter. Let $\cT\m{t}$ be given by \eqref{eq:T_Vander}.
		Then for each $l=1,\ldots,L$, we have
		\equ{
			\begin{split}
				 \cM(\m{X}_{:,l},\m{t})  = \begin{bmatrix} \overline{\m{A}} & \m{0}\\\m{0} & \m{A} \end{bmatrix} \begin{bmatrix} \m{B} & \diag(\m{B} \overline{\m{\Phi}_{:,l}}) \\ \diag(\m{B} \m{\Phi}_{:,l}) & \m{B} \end{bmatrix} \begin{bmatrix} \overline{\m{A}} & \m{0} \\\m{0} & \m{A} \end{bmatrix}^H  \ge \m{0}
			\end{split} \label{eq:decom}
		}
		and 
		\equ{
			\begin{split}
				\rank \sbra{\cM(\m{X}_{:,l},\m{t})} & = \rank \begin{bmatrix} \m{B} & \diag(\m{B} \overline{\m{\Phi}_{:,l}}) \\ \diag(\m{B} \m{\Phi}_{:,l}) & \m{B} \end{bmatrix} \\
				& = \rank\sbra{\m{B}} \le K, 
			\end{split} \label{eq:HTrank}
		}
		which holds since the Schur complement
		\equ{
			\m{B} - \diag(\m{B} \overline{\m{\Phi}_{:,l}}) \m{B}^{-1} \diag(\m{B} \m{\Phi}_{:,l}) = \m{0}. \label{eq:schur}
		}
		Consequently, we have $\m{X}\in \cS_{\text{HT}}$, and thus $\cS_0 \subset \cS_{\text{HT}}$.
		
		We next show $\cS_{\text{HT}} \subset \cS_0 $. For any $\m{X}\in \cS_{\text{HT}}$, we have $\rank \sbra{\cT\m{t}} \le K $ and $\cT \m{t} \ge \m{0}$. Using the Carath\'{e}odory-Fej\'{e}r theorem \cite[Theorem 11.5]{yang2018sparse}, there exist distinct $\lbra{\theta_k}^K_{k=1}$ and $\lbra{b_k \ge 0}^K_{k=1}$ such that $\cT\m{t}$ admits the decomposition as in \eqref{eq:T_Vander}. Moreover, it follows from the column inclusion property of PSD matrices that 
		\equ{
			\cH\m{X}_{:,l} \in \text{range}\sbra{\cT\m{t}}, \label{eq:range}
		} 
		for each $l = 1,\ldots,L$. Since $K<n$, it follows from the proof of \cite[Theorem 1]{wu2022maximum} that for a given $l$, $ \cH\m{X}_{:,l}$ must admit a Vandermonde decomposition 
		\equ{
			\cH\m{X}_{:,l} = \m{A}\bS^l\m{A}^T, \label{eq:H}
		} 
		where the diagonal matrix $\bS^l = \diag\sbra{\m{S}_{:,l}}$ with $\m{S}\in \bC^{K\times L}$. Consequently, we have that
		\equ{
			\begin{split}
				\cM(\m{X}_{:,l},\m{t}) = \begin{bmatrix} \overline{\m{A}} & \m{0}\\\m{0} & \m{A} \end{bmatrix} \begin{bmatrix} \m{B} & \overline{\bS^l} \\ \bS^l & \m{B} \end{bmatrix} \begin{bmatrix} \overline{\m{A}} & \m{0}\\\m{0} & \m{A} \end{bmatrix}^H.
			\end{split}  \label{eq:decom_2}
		}
		Applying the rank equality constraints in \eqref{eq:S_1}, we get that $\rank \sbra{\m{B}} = \rank \begin{bmatrix} \m{B} & \overline{\bS^l} \\ \bS^l & \m{B} \end{bmatrix}$, and thus the Schur complement 
		\equ{
			\m{B} - \overline{\bS^l} \m{B}^{-1} \bS^l = \m{0},
		}
		or equivalently,
		\equ{
			\abs{S_{k,l}} = b_k, \ k=1,\ldots,K, \, l=1,\ldots,L. \label{eq:mod_e}
		}
		This means that each row of $\m{S}$ is a CM signal. It follows from \eqref{eq:H} and \eqref{eq:mod_e} that $\m{X} = \m{A}(\m{\theta})\m{B}\m{\Phi} \in \cS_0$ with $\Phi_{k,l}=S_{k,l}/\abs{S_{k,l}}$, completing the proof of the first part.
		
		We then prove the second part. Since $b_k>0,k=1,\ldots,K$, we have $\rank\sbra{\cH \m{X}_{:,l}} = \rank\sbra{\diag(\m{B} \m{\Phi}_{:,l})} = K$. It then follows from \eqref{eq:range} and \eqref{eq:S_1} that $\rank\sbra{\cT\m{t}} = K$. Applying the Carath\'{e}odory-Fej\'{e}r theorem \cite[Theorem 11.5]{yang2018sparse}, there must exist a Vandermonde decomposition for $\cT\m{t}$, $\cT\m{t} = \m{A} \m{P} \m{A}^H$ with $\m{P} = \diag\sbra{\mbra{p_1,\ldots,p_K}}, p_k > 0$. Using $p_k > 0$, the rank inequality in \eqref{eq:S_1} and the decomposition in \eqref{eq:decom}, we have that
		\equ{
			\begin{split}
				\rank \sbra{\m{P}} = K & \ge \rank \sbra{\cM(\m{X}_{:,l},\m{t})} \\
				& = \rank \begin{bmatrix} \m{P} & \diag(\m{B} \overline{\m{\Phi}_{:,l}}) \\ \diag(\m{B} \m{\Phi}_{:,l}) & \m{P} \end{bmatrix},
			\end{split}
		}
		which yields that the Schur complement
		\equ{
			\m{P} - \diag(\m{B} \overline{\m{\Phi}_{:,l}}) \m{P}^{-1} \diag(\m{B} \m{\Phi}_{:,l}) = \m{0},
		}
		i.e.,
		\equ{
			p_k = b_k, \ k = 1,\ldots,K,
		}
		completing the proof.
	\end{proof}
	
	The assumption $K < n$ in Theorem \ref{thm:K<n} is to guarantee that the Hankel and Toeplitz matrices therein are low-rank and then admit the Vandermonde decompositions. But we point out that $K<n$ is not a limit. The case of $K\ge n$ will be tackled in Section \ref{sec:extensions}.
	
	Theorem \ref{thm:K<n} demonstrates that the structures in $\cS_0$ are fully characterized by  $\cS_{\text{HT}}$. Based on this, in the noiseless case, we reformulate the problem in \eqref{eq:noiseless} as the following rank-constrained Hankel-Toeplitz matrix-based feasibility problem:
	\equ{
		\begin{split}
			& \text{find} \ \m{t}, \st \rank\sbra{\cT\m{t}} = \rank\sbra{\cM(\m{X}^\star_{:,l},\m{t})}\leq K, \\
			& \qquad \qquad \qquad \quad \  \cM(\m{X}^\star_{:,l},\m{t})\geq \m{0},\;  l=1,\ldots,L.
		\end{split} \label{eq:fea1}
	}
	It follows from Theorem \ref{thm:K<n} that the optimal solution of \eqref{eq:fea1} is uniquely given by \eqref{eq:Tt}.
	Once we find the solution $\m{t}$, the DOAs $\m{\theta}$ and the moduli $\m{b}$ are extracted from $\cT\m{t}$ by finding its Vandermonde decomposition in \eqref{eq:Tt} via subspace methods \cite{schmidt1986multiple,barabell1983improving,roy1989esprit}, and then the phases $\m{\Phi}$ are retrieved from $\m{X}^\star$.
	
	In the noisy case, the minimization problem in \eqref{eq:p_S0_CM} is written equivalently as the following rank-constrained Hankel-Toeplitz matrix-based recovery problem:
	\equ{
		\begin{split}
			& \min_{\m{X},\m{t}} \frobn{\m{Y} - \m{X}}^2, \\
			& \st \rank\sbra{\cT\m{t}} = \rank\sbra{\cM(\m{X}_{:,l},\m{t})} \leq K, \\
			& \qquad \qquad \ \; \cM(\m{X}_{:,l},\m{t})\geq \m{0},\;  l=1,\ldots,L.
		\end{split} \label{eq:p_S1}
	}
	Suppose that an optimal solution to the NLS problem in \eqref{eq:NLS_CM} is given by $\set{\breve{\m{\theta}},\breve{\m{b}},\breve{\m{\Phi}} }$. We have $\breve{b}_k \neq 0$ almost surely\footnote{If $\breve{b}_k =  0$, then $\m{Y}$ is composed exactly of $K-1$ or fewer source signals, which occurs with probability zero in the presence of random noise.}. An optimal solution to \eqref{eq:p_S1} is then given by
	\equ{ 
		\begin{split}
			\breve{\m{X}} &= \sum_{k=1}^K\m{a}(\breve{\theta}_k) \breve{b}_k \breve{\m{\Phi}}_{k,:}, \\
			\cT\breve{\m{t}} &= \m{A}_n(\breve{\m{\theta}})\diag(\breve{\m{b}})\m{A}_n^H
			(\breve{\m{\theta}}),
		\end{split}
	}
	where $\breve{\m{t}}$ is uniquely determined by $\breve{\m{X}}$.
	Once the problem in \eqref{eq:p_S1} is solved, the MLE of parameters can be easily extracted from $\cT \m{t}$ and $\m{X}$ by computing the Vandermonde decomposition of $\cT \m{t}$ and a least squares method. 
	
	\subsection{Relaxation of Rank Equalities} \label{sec:prop}
	It is difficult to tackle explicitly the rank equality constraints 
	\equ{
		\rank\sbra{\cT\m{t}} = \rank\sbra{\cM(\m{X}_{:,l},\m{t})}, l = 1,\ldots,L \label{eq:rankeq}
	} 
	in \eqref{eq:fea1} and \eqref{eq:p_S1} in an algorithm.
	A natural idea is to drop the constraints above by relaxing $\cS_{\text{HT}}$ to the following set:
	\equ{
		\begin{split}
			\cS'_{\text{HT}} \! & = \! \lbra{{\m{X} \! : \rank\sbra{\cM(\m{X}_{:,l},\m{t})}\leq K, \cM(\m{X}_{:,l},\m{t})\geq \m{0},} \right. \\
				& \qquad \qquad \qquad \qquad \qquad  \left.  
				{ l=1,\ldots,L, \text{for some } \m{t} \in \bC^{n}} } \\
			& = \! \lbra{ \m{X} \! :  \cM(\m{X}_{:,l},\m{t}) \in \bS^+_K,  l=1,\ldots,L,  \text{for some } \m{t} \in \bC^{n} }  \label{eq:S_2}
		\end{split}
	}
	where $\bS^+_K$ is the set of positive semidefinite matrices of rank no greater than $K$.
	The resulting optimization problems are given by:
	\equ{
		\text{find} \ \m{t}, \ \st \cM(\m{X}^\star_{:,l},\m{t}) \in \bS^+_K, \  l=1,\ldots,L, \label{eq:p_S3}
	}
	in the noiseless case and 
	\equ{
		\begin{split}
			\min_{\m{X},\m{t}} \frobn{\m{Y} - \m{X}}^2,
			& \st \cM(\m{X}_{:,l},\m{t}) \in \bS^+_K, \  l=1,\ldots,L.
		\end{split} \label{eq:p_S2}
	}
	in the noisy case.
	We have the following result.
	\begin{prop} \label{prop:1}
		If $K=1$, we have $\cS'_{\text{HT}} = \cS_{\text{HT}} = \cS_0$. If $K>1$, then
		\equ{
			\cS'_{\text{HT}} = \cS'_0 \triangleq \cS_0 \cup \cS''_0,
		}
		where 
		\equ{
			\begin{split}
				\cS''_0 
				= & \lbra{
					{ \m{A}' \sbra{\m{\theta}} \m{B} \m{\Psi}: \abs{\Psi_{k,l}}\leq 1, \; k = 1,\ldots,K', \; l = 1,\ldots,L,   } \right. \\
					& \left.  
					{ \qquad \qquad \quad \quad  \abs{ \lbra{k: \abs{\Psi_{k,l}}<1 } } \le \min \lbra{K', K - K'}, K' < K} 
				}
			\end{split} \label{eq:S1}
		}with the $N \times K'$ Vandermonde matrix $\m{A}' \sbra{\m{\theta}} $, $\m{B}=\diag\sbra{\m{b}} \in \bR^{K' \times K'}$ with $\lbra{b_k > 0}$ and $\m{\Psi} \in \bC^{K' \times L}$.
	\end{prop}
	
	\begin{proof}
		See Appendix \ref{append:1}.
	\end{proof}
	
	It is implied by Proposition \ref{prop:1} that the relaxed set $\cS'_{\text{HT}}$ can be viewed as a relaxation of the set $\cS_0$ by allowing appearance of the elements in $\cS''_0$.
	To better illustrate the set $\cS'_{\text{HT}}$, we present in Fig.~\ref{fig:table} a simple example in the case of $K = L = 4$. According to Proposition \ref{prop:1}, the number $K'$ of sources in $\cS''_0$ can be $1$, $2$ and $3$. As $K' = 3$, up to $K - K' = 1$ entry in each column of $\m{\Psi}$ can lie in the unit circle, with the others on the unit circle. Similar arguments hold in the case of $K' =1,2$. 
	
	\begin{figure}
		\centerline{\includegraphics[width=11cm]{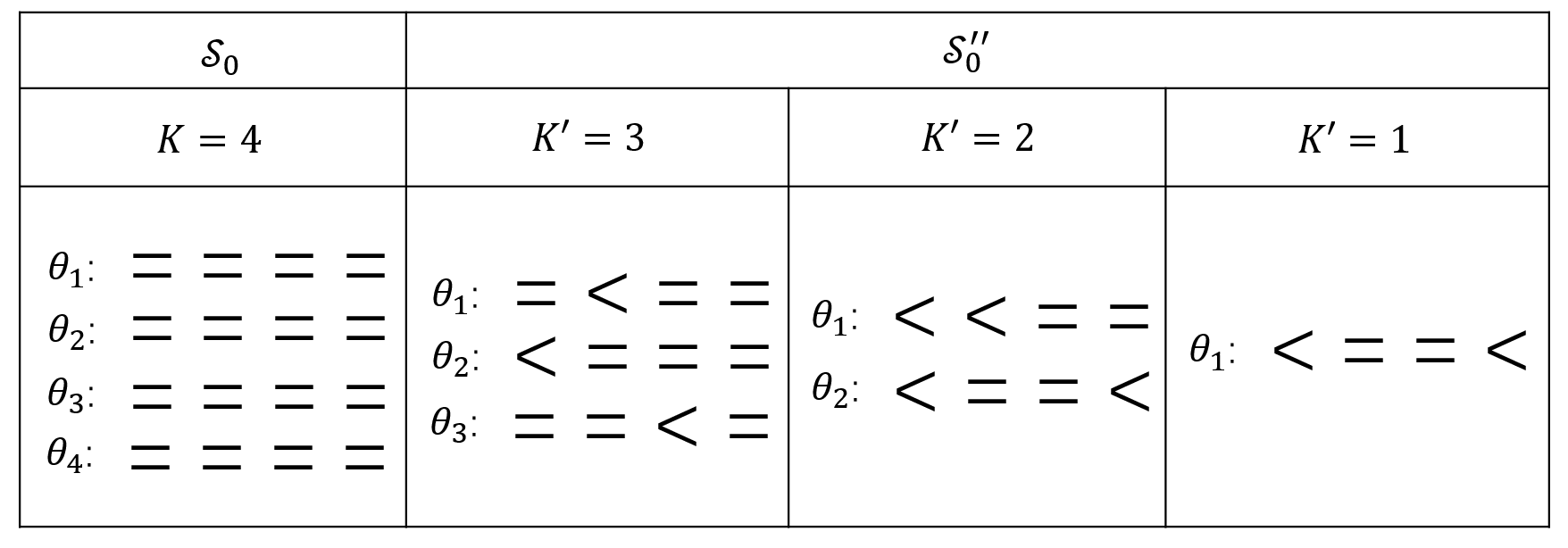}}
		\caption{An example with $K=L=4$ where ``$=$'' denotes $\abs{\Phi_{k,l}} = 1$ in $\cS_0$ or $\abs{\Psi_{k,l}} = 1$ in $\cS''_0$, and ``$<$'' denotes $\abs{\Psi_{k,l}} < 1$ in $\cS''_0$.}
		\label{fig:table}
	\end{figure}
	
	It is seen that the true signal set $\cS_0$ has $2K+KL$ degrees of freedom (DoFs), i.e., $K$ DOAs, $K$ moduli and $KL$ phases, while the set $\cS''_0$ has $2K' +KL$ DoFs at most, i.e., $K'$ DOAs, $K'$ moduli and up to $K L$ variables in $\m{\Psi}$.  Intuitively, given the prior knowledge that the true signal is in $\cS_0$, it is unlikely that the estimated signal belongs to the set $\cS''_0$ at the cost of reduced number of sources and DoFs. This implies that the proposed relaxation is tight and does not alter the optimal solution in general, which will be validated by numerical results in Section \ref{sec:sim}.
	
	\section{Algorithms} \label{sec:alg}
	Both the feasibility problem in \eqref{eq:p_S3} and the minimization problem in \eqref{eq:p_S2} are nonconvex and challenging to solve since they are rank-constrained. Interestingly, the ADMM algorithm has shown good convergence and performance for nonconvex rank-constrained problems \cite{andersson2014new,boyd2011distributed,diamond2018general,wu2022maximum}. Some theoretical results have also been developed in \cite{xu2012alternating,jiang2014alternating,wang2019global}. Therefore, we apply ADMM to solving the two rank-constrained problems in this paper.
	In particular, both the problems in \eqref{eq:p_S3} and  \eqref{eq:p_S2} can be rewritten as:
	\equ{
		\min_{\m{X},\m{t}} f\sbra{\m{X}}, \st \cM(\m{X}_{:,l},\m{t}) \in \bS^+_K, \; l=1,\ldots,L, \label{eq:general}
	}
	where $f\sbra{\m{X}}$ denotes a cost function of $\m{X}$. To apply the ADMM, we introduce a series of auxiliary Hermitian matrix variables $\lbra{\m{Q}^l}_{l=1}^L$ and rewrite \eqref{eq:general} as:
	\equ{ 
		\begin{split}
			 \min_{\m{X},\m{t}, \lbra{\m{Q}^l}_{l=1}^L } f\sbra{\m{X}} + \sum^L_{l=1}\mathbb{I}_{\bS^+_K}\sbra{\m{Q}^l},   \st \m{Q}^l = \cM(\m{X}_{:,l},\m{t}), \; l = 1, \ldots, L,
		\end{split} \label{eq:ADMM_fun}
	}
	where $\mathbb{I}_{\mathscr{S}}(\cdot)$ is the indicator function: $\mathbb{I}_{\mathscr{S}}(\m{W})=0$ if $\m{W}\in \mathscr{S}$ or $+\infty$ otherwise and the set $\mathscr{S}=\bS^+_K$ here. For notational simplicity, we write $\cM(\m{X}_{:,l},\m{t})$ as $\cM^l$ hereafter if needed.
	The augmented Lagrangian function is given by
	\equ{\begin{split}
			& \cL\sbra{\m{X}, \m{t},\lbra{\m{Q}^l}^L_{l=1},\lbra{\m{\Lambda}^l}^L_{l=1} } \\
			& = f\sbra{\m{X}} + \sum^L_{l=1}\mathbb{I}_{\bS^+_K}\sbra{\m{Q}^l} + \sum^L_{l=1} \langle \m{\Lambda}^l,\m{Q}^l -\cM^l \rangle_{\bR}  + \frac{\rho}{2} \sum^L_{l=1} \frobn{\m{Q}^l-\cM^l}^2, \\
	\end{split}}
	where $\rho>0$ is a penalty parameter and $\lbra{\m{\Lambda}^l}^L_{l=1}$ is a series of Hermitian Lagrangian multiplier. Assume that at iteration $j$ we have computed $\cM^l_{j}$ and $\m{\Lambda}^l_j$ for $l=1,\ldots,L$, the $(j+1)$-th iteration of ADMM is given by
	{
		\lentwo\equa{
			\lbra{\m{Q}_{j+1}^l} &=& \argmin_{	\lbra{\m{Q}^l} \in \bS^+_K} \sum^L_{l=1} \frobn{ \m{Q}^l -\cM^l_{j}+{\rho^{-1}}\m{\Lambda}^l_j }^2, \label{eq:1} \\
			\sbra{\m{X}_{j+1},\m{t}_{j+1}} & = & \argmin_{\m{X},\m{t}} f\sbra{\m{X}} \nonumber \\ 
			&& \; + \frac{\rho}{2} \sum^L_{l=1} \frobn{\m{Q}^l_{j+1} -\cM^l +{\rho^{-1}}\m{\Lambda}^l_j}^2, \label{eq:2} \\
			\m{\Lambda}_{j+1}^l &=& \m{\Lambda}_j^l + \rho \sbra{\m{Q}_{j+1}^l - \cM_{j+1}^l}, \; l=1,\ldots,L. \label{eq:Lambda2} 
		}
	}For the first subproblem in \eqref{eq:1}, we have the updates \cite{dax2014low}
	\equ{
		\m{Q}_{j+1}^l = \cP_{\bS^+_K}\sbra{\cM_{j}^l-{\rho^{-1}}\m{\Lambda}^l_j}, \label{eq:Q_HT}
	}
	where the projection $\cP_{\bS^+_K}\sbra{\cdot}$ is obtained as the truncated eigen-decomposition of the Hermitian matrix argument by setting all but the largest $K$ (or less) positive eigenvalues to zero. 
	
	For the second subproblem in \eqref{eq:2}, the variables $\m{X}$ and $\m{t}$ are separable.
	Denote a Hermitian matrix $\m{W}^l = \m{Q}^l_{j+1}+{\rho^{-1}}\m{\Lambda}^l_j$ and write $\m{W}^l = \begin{bmatrix} \m{W}^l_{1} & \sbra{\m{W}^l_{2}}^H \\ \m{W}^l_{2} & \m{W}^l_{3} \end{bmatrix}$ as a block matrix like $\cM^l$. 
	The closed-form solution to $\m{t}$ is given by
	\equ{
		\m{t}_{j+1} =  \frac{1}{2L} \m{D}_{\cT}^{-1} \cT^H \sum^L_{l=1}\sbra{\overline{\m{W}^l_{1}}+\m{W}_{3}^l}, \label{eq:t_HTs}
	}
	where $\cT^H$ denotes the adjoint of the Hermitian Toeplitz operator $\cT$ by mapping an $n \times n $ Hermitian matrix $\m{C}$ to an $n$-dimensional vector $\cT^H\m{C}=\set{\sum_{i-j=a-1}C_{ij}}_{a=1}^{n}$ and $\m{D}_{\cT}=\diag\sbra{\mbra{n,n-1,\dots,1}}$. The derivation of \eqref{eq:t_HTs} is given in Appendix \ref{append:2}. The solution to $\m{X}$ depends on the cost function $f\sbra{\m{X}}$ and will be given later.
	
	The ADMM iterations are terminated if primal residual $\sum^L_{l=1} \frobn{\m{R}^{l}_{j+1,p}}$ and dual
	residual $ \sum^L_{l=1} \frobn{\m{R}^{l}_{j+1,d}}$ computed with
	\equ{
		\m{R}^{l}_{j+1,p} = \m{Q}_{j+1}^l - \cM^l_{j+1} \quad \text{and} \quad
		\m{R}^{l}_{j+1,d} = \rho \sbra{  \cM^l_{j+1} -  \cM^l_j }, \nonumber
	}
	respectively, for $l=1,\ldots,L$, are sufficiently small or a maximum number of iterations is reached, see \cite{boyd2011distributed}.
	Once ADMM has converged, we compute the Vandermonde decomposition in \eqref{eq:Tt} of the Toeplitz submatrix in $\cM^l$ using ESPRIT, from which the estimated DOAs $\widehat{\m{\theta}}$ are obtained. 
	
	\subsection{ADMM for the Feasibility Problem in \eqref{eq:p_S3}}
	We make the substitution
	\equ{
		f\sbra{\m{X}} = \mathbb{I}_{\lbra{\m{0}}}\sbra{\m{X}-\m{X}^\star},
	}
	which fixes the variable $\m{X}$ on $\m{X}^\star$. For the subproblem in \eqref{eq:2}, it has the following closed-form solution to $\m{X}$:
	\equ{
		\m{X}_{j+1} = \m{X}^\star. \label{eq:fix}
	} 
	
	\subsection{ADMM for the Minimization Problem in \eqref{eq:p_S2}} 
	We make the substitution
	\equ{
		f\sbra{\m{X}} = \frobn{\m{X}-\m{Y}}^2. \label{eq:loss}
	}
	For the subproblem in \eqref{eq:2}, it has the following closed-form solution to $\m{X}$:
	\equ{
		\sbra{\m{X}_{:,l}}_{j+1} = \sbra{\m{I} + \rho \m{D}_{\cH}}^{-1} \sbra{\m{Y}_{:,l} + \rho \cH^H\m{W}_2^l}, \label{eq:x_HT}
	}
	for each $l=1,\ldots,L$, 
	where $\mathcal{H}^H$ denotes the adjoint of the Hankel operator $\mathcal{H}$ by mapping an $n \times n$ matrix $\m{C}$ to an $N$-dimensional vector with $\mathcal{H}^H\m{C}=\set{\sum_{i+j-1=a}C_{ij}}_{a=1}^{N}$ and $\m{D}_\cH = \diag\sbra{\mbra{1,2,\dots,n,n-1,\dots,1}}$. The derivation of \eqref{eq:x_HT} is given in Appendix \ref{append:2}.
	
	\begin{algorithm}
		\caption{\emph{S}tructured \emph{ma}trix \emph{r}ecovery \emph{t}echnique (SMART) for CM DOA estimation.}
		\begin{algorithmic}[1]
			\REQUIRE Observation $\m{X}^\star$ (or $\m{Y}$), source number $K$.
			\ENSURE Estimates of parameters $\lbra{\m{\theta},\m{b},\m{\Phi}}$.
			\STATE Initialize $\cM^l_0$, $\m{\Lambda}^l_0$. 
			\WHILE{not converged}
			\STATE Conduct the updates in \eqref{eq:Q_HT}, \eqref{eq:t_HTs}, \eqref{eq:fix} (or \eqref{eq:x_HT}), and \eqref{eq:Lambda2} one after one.
			\ENDWHILE
			\STATE Estimate $\lbra{\m{\theta},\m{b}}$ by computing the decomposition in \eqref{eq:Tt} using ESPRIT and $\m{\Phi}$ from $\m{X}$.
		\end{algorithmic} \label{alg:alg_HT}
	\end{algorithm}
	
	\subsection{Complexity and Convergence}
	We summarize the ADMM algorithm for the feasibility problem in \eqref{eq:p_S3} and the minimization problem in \eqref{eq:p_S2} in Algorithm \ref{alg:alg_HT}. We call the proposed approach as structured matrix recovery technique (SMART) due to the key Hankel-Toeplitz matrix optimization model. 
	The computations of SMART are dominated by the projection $\cP_{\bS^+_K}$ in \eqref{eq:Q_HT} that is computed by the truncated eigen-decomposition of which the computational complexity is $\mathcal{O}(N^2K)$ \cite[Section 3.3.2]{halko2011finding}, so the total complexity per iteration is $\mathcal{O}(N^2KL)$. It is worth noting that the updates in \eqref{eq:Q_HT} involving $\cP_{\bS^+_K}$ are independent and can be computed in parallel for each snapshot. Therefore, the computational complexity per iteration can be reduced to $\mathcal{O}(N^2K+N^2L)$ by using parallel computing where $N^2L$ exists due to the updates of $\m{t}$ in \eqref{eq:t_HTs}. Practical implementation of ADMM for large $L$ will depend on the specific parallel architecture and programming environment used. 

		For convex optimization problems, the global convergence of the ADMM has been extensively studied and understood \cite{boyd2011distributed}; nevertheless, the convergence for nonconvex problems is still ongoing \cite{jiang2014alternating,li2015global,wang2019global}. 
		Inspired by \cite{jiang2014alternating,wu2022maximum}, we provide convergence analysis for the proposed SMART. We denote $\m{z}=\lbra{\m{X},\m{t}}$, $f(\m{z})=f(\m{X})$, and $\cA \m{z} = \lbra{ \cM \sbra{\m{X}_{:,l},\m{t}} }^L_{l=1} $ that is linear in (the real and complex parts of) $\m{z}$. The problem in \eqref{eq:ADMM_fun} is rewritten as
	\equ{
		\min_{\m{z}, \lbra{\m{Q}^l}} f\sbra{\m{z}} + \sum^L_{l=1} \mathbb{I}_{\bS_K^+}\sbra{\m{Q}^l}, 
		\st \lbra{\m{Q}^l} = \cA\m{z}. \label{eq:P_conver}
	}
	We have the following theorem, of which the detailed proof is similar to that of \cite[Theorem 2]{wu2022maximum} and will be omitted.
	
		\begin{thm} \label{thm:Convergence}
		Let $\lbra{\lbra{\m{Q}^l_j}, \m{z}_j, \lbra{\m{\Lambda}^l_j}}$ be a sequence generated by SMART. Assume that
		\equ{
			\lim_{j \to \infty}  \frobn{ \m{z}_{j+1} - \m{z}_j }^2 +
			\sum^L_{l=1} \frobn{ \m{\Lambda}^l_{j+1} - \m{\Lambda}^l_j }^2 = 0. \label{eq:assum}
		}
		Then for any limit point $\lbra{\lbra{\m{Q}^l_*}, \m{z}_*, \lbra{\m{\Lambda}^l_*}}$, $\lbra{\m{z}_*, \lbra{\m{Q}^l_*}}$ is a stationary point of \eqref{eq:P_conver}, i.e.,
		\equ{
			\begin{split}
				& \m{0} \in \partial \mathbb{I}_{\bS_K^+}\sbra{\m{Q}^l_*} + \m{\Lambda}^l_*, \; l = 1,\ldots,L, \\
				& \nabla f \sbra{\m{z}_*} = \cA^H \lbra{\m{\Lambda}_*^l}, \\
				& \lbra{\m{Q}^l_*} = \cA \m{z}_*,
			\end{split} \label{eq:stationary}
		}
		where $\cA^H$ is the adjoint operator of $\cA$ and $\partial \mathbb{I}_{\bS_K^+}$ is the general subgradient \cite[Definition 8.3]{rockafellar2009variational}.
	\end{thm}
	
		Theorem \ref{thm:Convergence} shows that if the solution sequence produced by SMART converges, then it converges to a stationary point. Stronger convergence analyses in \cite{li2015global,wang2019global} are not applicable for SMART because the mappings $\cM$ and $\cA$ are not surjective in our problem. Good convergence of SMART will be confirmed via extensive numerical simulations in Section \ref{sec:sim}.
	
	\section{Extensions} \label{sec:extensions}
	In this section, we extend the proposed SMART including the Hankel-Toeplitz optimization model and the corresponding ADMM algorithm to the large number of sources case with $K\ge n$ and the SLA case.
	
	\subsection{The Case of $K\ge n$} \label{sec:K>=n}
	Recall that $n=\sbra{N+1}/2$. When $K \ge n$, we define an enlarged virtual dataset $\m{X}' \in \bC^{N'\times L}$ with $N' > N$ and $\m{X}^\star$ can be regarded as incomplete observations of $\m{X}'$. In particular, we have
	\equ{
		\m{X}'_{\mbra{N}} \triangleq \cP_{\mbra{N}}\sbra{\m{X}'} = \begin{bmatrix} \m{X}^\star \\ \m{0} \end{bmatrix}, \label{eq:consistent}
	}
	where $\mbra{N} = \lbra{1,\ldots,N}$. As in the case of $K<n$, we assume that $N'$ is odd and let $n' = \sbra{N'+1}/2$. Denote the variable $\m{X}\in \bC^{N'\times L}$. To guarantee that the $n'\times n'$ submatrices $\cT\m{t}$ and $\cH\m{X}_{:,l}$ in $\cM(\m{X}_{:,l},\m{t})$ are low-rank, we let
	\equ{
		N' \geq 2K+1. \label{eq:N'}
	}
	Moreover, we redefine the sets $\cS_0$ in \eqref{eq:S0}, $\cS_{\text{HT}}$ in \eqref{eq:S_1}, and $\cS'_{\text{HT}}$ in \eqref{eq:S_2} by changing the dimension from $N$ to $N'$.
	
	The condition in \eqref{eq:N'} means that $K< n'$, and thus $\cH \m{X}_{:,l} \in \bC^{n'\times n'}$ and $\cT\m{t} \in \bC^{n'\times n'}$ admit the Vandermonde decomposition as in \eqref{eq:H} and \eqref{eq:T_Vander}, respectively. Following the same steps as in the proof of Theorem \ref{thm:K<n}, we have $\cS_{\text{HT}} = \cS_0$.
	Since the rank equality constraints therein are difficult to be tackled, we relax $\cS_{\text{HT}}$ to $\cS'_{\text{HT}}$ and consider the following optimization problems
	\equ{
		\begin{split}
			\!\!\!\!\! \text{find}  \lbra{\m{X}, \m{t}}, \st
			& \cM(\m{X}_{:,l},\m{t})\in \bS^+_K, \; l=1,\ldots,L, \\
			& \m{X}_{\mbra{N}} = \m{X}'_{\mbra{N}},
		\end{split} \label{eq:p_K_n}
	} 
	in the absence of noise, and
	\equ{
		\begin{split}
			& \min_{\m{X},\m{t}} \frobn{ \m{Y}' - \m{X}_{\mbra{N} } }^2,  \st \cM(\m{X}_{:,l},\m{t})\in \bS^+_K, \; l=1,\ldots,L, \label{eq:p_K_n2}
		\end{split}
	}
	where $\m{Y}' = \begin{bmatrix} \m{Y} \\ \m{0} \end{bmatrix} \in \bC^{N'\times L}$ in the presence of noise.
	It has been shown in Section \ref{sec:prop} that the above relaxation is reasonably tight.
	
	As in the case of $K<n$ in Section \ref{sec:alg}, the ADMM algorithm can be similarly derived for the problem in \eqref{eq:p_K_n} by making the substitution $f\sbra{\m{X}} = \mathbb{I}_{\lbra{\m{0}}} \sbra{ \m{X}_{\mbra{N}} - \m{X}'_{\mbra{N}} }$ in \eqref{eq:general}. The only difference is the update of $\m{X}$:
	\equ{
		\sbra{X_{q,l}}_{j+1} = 
		\begin{cases}
			X_{q,l}^\star \quad & q = 1,\ldots,N,  \\
			\mbra{\m{D}_{\cH}^{-1}\cH^H\m{W}^l_{2}}_q & q = N+1,\ldots,N',
		\end{cases}
	}
	for $l=1,\ldots,L$.
	
	For the problem in \eqref{eq:p_K_n2}, we make the substitution $f\sbra{\m{X}} = \frobn{ \m{Y}' - \m{X}_{\mbra{N} } }^2$ in \eqref{eq:general}. The only difference also is the update of $\m{X}$:
	\equ{
		\sbra{\m{X}_{:,l}}_{j+1} =  \sbra{\diag\sbra{\cP_{\mbra{N}}\sbra{\m{1}}} + \rho \m{D}_{\cH}}^{-1} \sbra{\m{Y}'_{:,l} + \rho \cH^H\m{W}^l_2},
	}
	for $l=1,\ldots,L$, where $\m{1}$ is the $N' \times 1$ vector of ones.
	
	It has been shown in \cite{wax1992unique,williams1992resolving,valaee1994alternative} that the condition $K<2N-2$ is sufficient to uniquely identify the CM signal sources using a ULA, implying that $K$ can be greater than $N$ and it is possible to localize more than $N$ source signals from $N$ sensors. Since the proposed SMART makes full use of the CM and Vandermonde structures, it is expected that more than $N$ sources can be localized by solving the feasibility problem in \eqref{eq:p_K_n}, which will be verified numerically in Section \ref{sec:sim}.
	
	\subsection{The SLA Case} \label{sec:SLA}
	The proposed SMART can be easily generalized to the SLA case in \eqref{eq:SLA_data}. We consider the problem
	\equ{
		\begin{split}
			\!\!\!\!\! \text{find}  \lbra{\m{X}, \m{t}}, \st 
			& \cM(\m{X}_{:,l},\m{t})\in \bS^+_K, \; l=1,\ldots,L, \\
			& \m{X}_{ \Omega} = \m{X}^\star_{ \Omega},
		\end{split}
	} 
	in the noiseless case and the problem
	\equ{
		\begin{split}
			& \min_{\m{X},\m{t}} \frobn{\m{Y}_{\Omega} - \m{X}_{ \Omega}}^2, \st \cM(\m{X}_{:,l},\m{t})\in \bS^+_K, \; l=1,\ldots,L,
		\end{split}
	}
	in the noisy case. The ADMM algorithm can be derived similarly as in the case of ULA and the case of $K\ge n$ in the last subsection and we omit the details. 
	It is worthy noting that the case of $K\ge n$ can be regarded as a special SLA in which the number of complete sensors is $N'$ and the observation set $\Omega$ is $\mbra{N}$.	
	
	\begin{figure*}
		\centerline{\includegraphics[width=17cm]{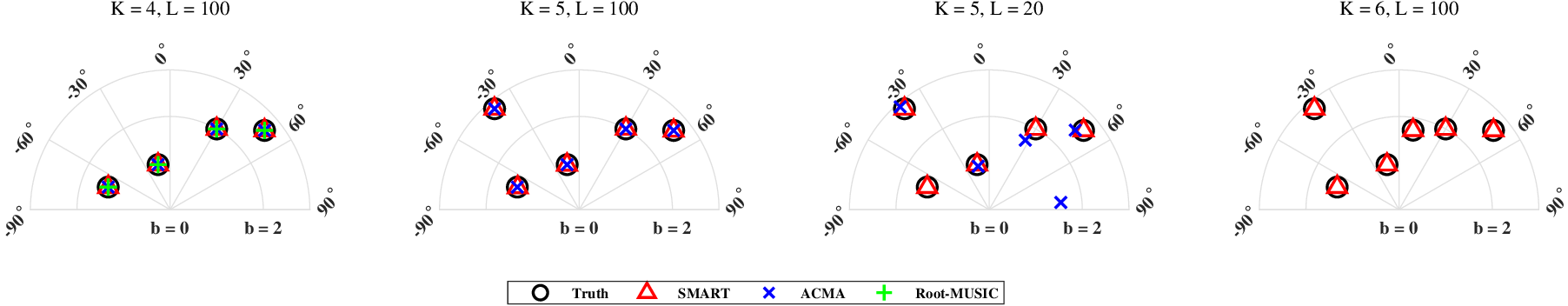}}
		\caption{Comparison of source localization between SMART, ACMA-based method and Root-MUSIC for a ULA with $N = 5$ elements.}
		\label{fig:ID_polar}
	\end{figure*}
	
	\section{Numerical Results} \label{sec:sim}
	
	\subsection{The Noiseless Case}	
	In this subsection, we present numerical results to illustrate the performance of the proposed SMART in the absence of noise. In ADMM, all matrix variables are initialized with zero, the penalty parameter $\rho$ is fixed on one and the algorithm is terminated if the maximum number of iterations $10^4$ is exhausted.
	
	\begin{figure}[htbp]
		\centerline{\includegraphics[width=11cm]{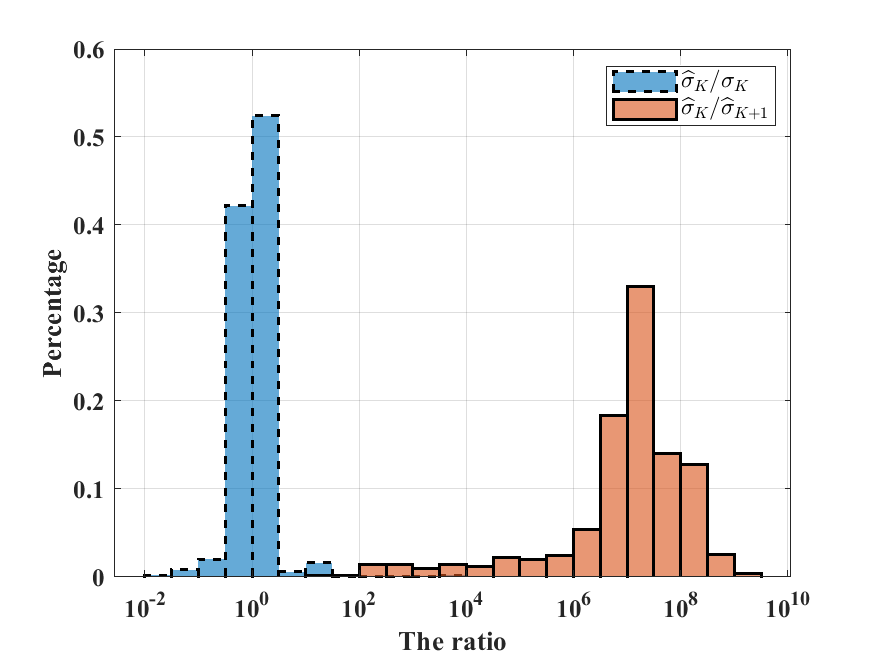}}
		\caption{Histogram of the ratio between singular values of the truth $\cT\m{t}$ and the solution $\cT\widehat{\m{t}}$ in SMART.}
		\label{fig:Histogram}
	\end{figure}
	
	\begin{figure*}[htb] 
		\centering
		\subfigure[] {\includegraphics[width=8cm]{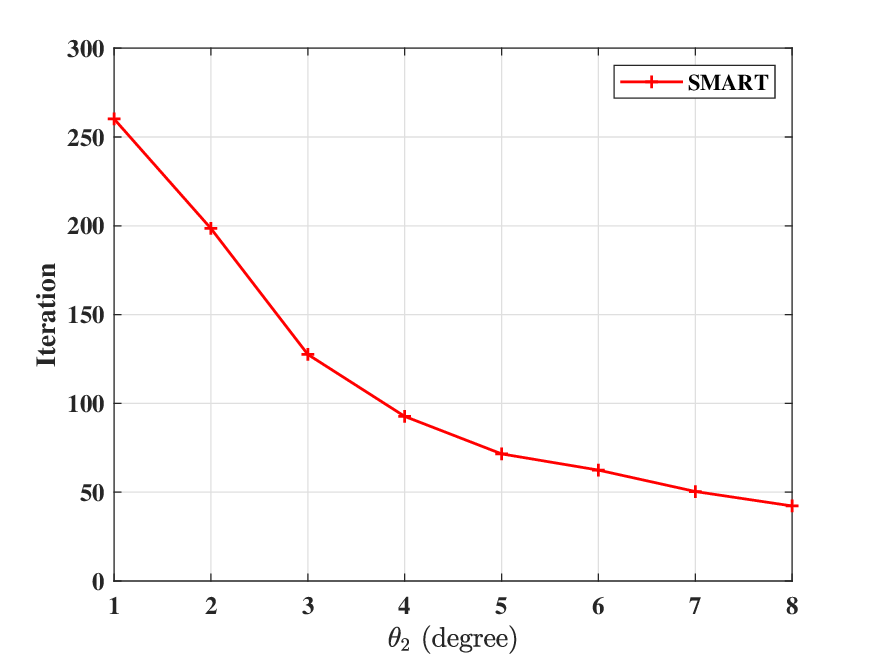}}
		\subfigure[] {\includegraphics[width=8cm]{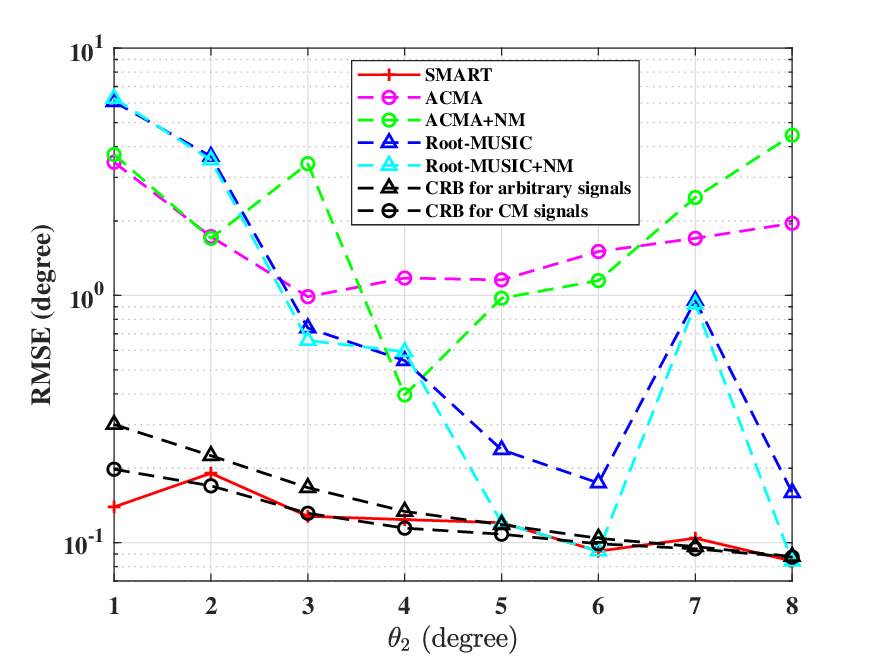}}
		
		\subfigure[] {\includegraphics[width=8cm]{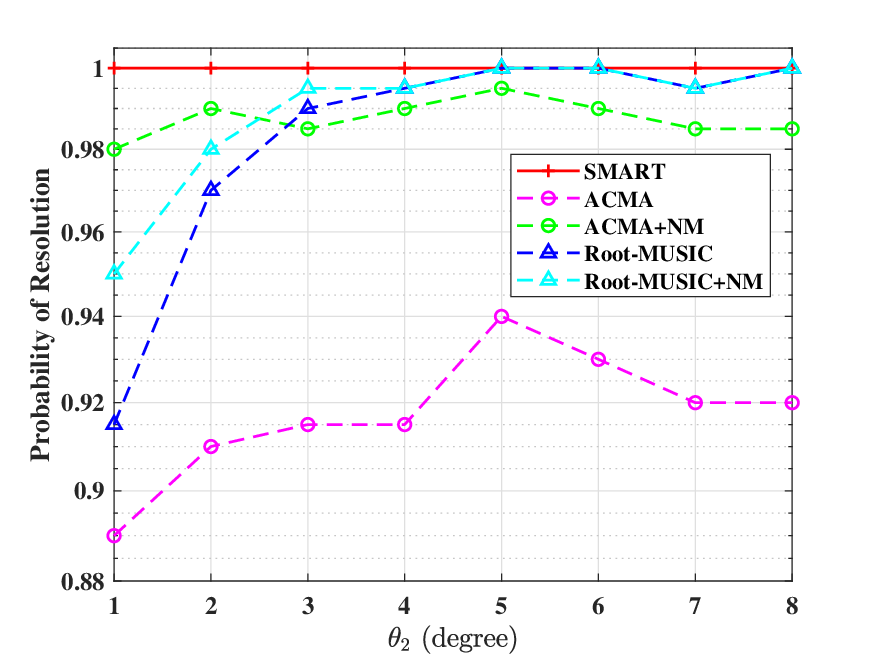}}
		\subfigure[] {\includegraphics[width=8cm]{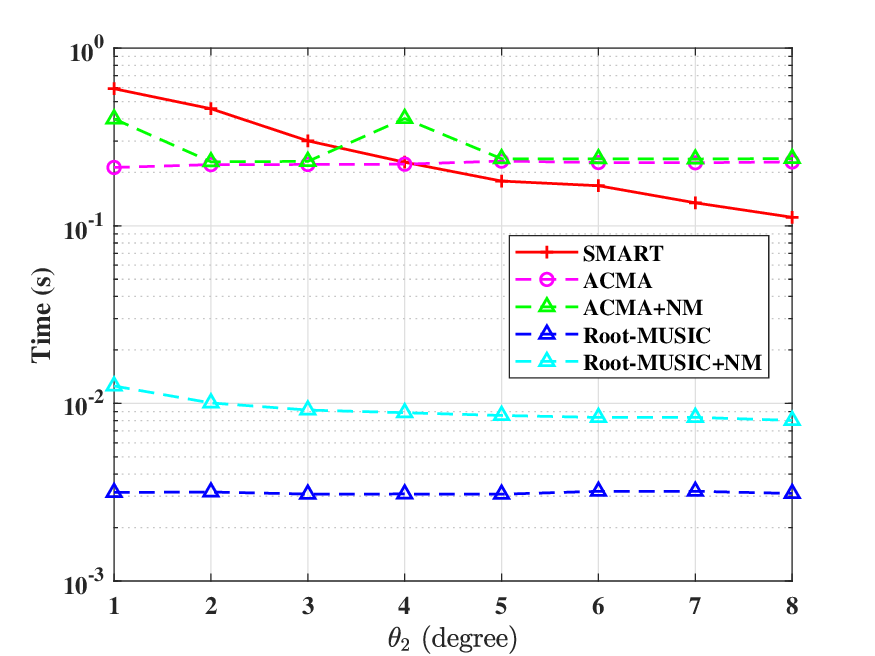}}
		\caption{(a) Iteration number; (b) RMSE in degree; (c) The probability of resolution; (d) The running time of algorithms versus the separation $\Delta_{\theta}$ for a ULA with $N = 15$ elements, SNR$= 20$dB and snapshots $L = 3$.}
		\label{fig:RMSEvsResolution}
	\end{figure*}
	
	\emph{Experiment 1: Source Localization.} We evaluate the performance of source localization of SMART via comparing it with ACMA \cite{leshem1999direction} and Root-MUSIC \cite{barabell1983improving}. A ULA with $N=5$ sensors is considered. A number of $K$ DOAs $\lbra{\theta_k}$ are chosen from the set $\lbra{-70^{\circ},-40^{\circ},-15^{\circ},10^{\circ},30^{\circ},50^{\circ}}$, and CM source signals are generated with powers $\lbra{b_k^2}$ chosen from the set $\lbra{2,8,1,3,4,7}$ and with random phases. We consider $K\in \set{4,5,6}$. Since $K>n=(N+1)/2=3$, we solve the problem in \eqref{eq:p_K_n} and select $N'=2K+1$ in SMART to satisfy the condition in \eqref{eq:N'}.
	As shown in Fig.~\ref{fig:ID_polar}, all methods succeed to estimate the DOA and the modulus of signal when the source number $K=4<N$ and the snapshot number $L=100$. As $K$ is greater than $4$, Root-MUSIC absolutely fails and is omitted since it does not utilize the CM structure and can localize at most $N-1=4$ sources. When $L$ decreases to $20$, ACMA fails since one of the sufficient conditions $K< \sqrt{L}$ of ACMA in \cite{van1996analytical} is not satisfied. When $K$ increases to $6$, ACMA is not applicable since its necessary condition $K\le N$ is not met. The proposed SMART converges and succeeds to localize sources even when $K=6>N$. 
	
	\subsection{The Noisy Case}
	In this subsection, we investigate the statistical efficiency of SMART in the presence of noise. The methods that we use for comparison include Root-MUSIC \cite{barabell1983improving}, ACMA \cite{van1996analytical,leshem1999direction}, and Newton's method (NM) \cite{leshem2000maximum}, which are representatives of the subspace methods (without considering the CM structure), the two-step methods, and the parameter-domain MLE methods, respectively, for CM DOA estimation. The CRBs for CM and arbitrary signals \cite{leshem1999direction,stoica1989music} are presented as benchmarks. For SMART, the penalty parameter $\rho$ in ADMM is initialized by $\rho_0=1/\sqrt{N+L}$ and is adaptively updated as in \cite[Section 3.4.1]{boyd2011distributed} to accelerate convergence. We initialize all matrix variables with zero. The ADMM will be terminated if a maximum number of $3000$ iterations are reached. SMART is implemented in MATLAB without using parallel computing. For ACMA, one-dimensional grid search with a uniform grid of size $10^{4}$ is used to estimate the DOAs from the estimated steering matrix as in \cite{leshem1999direction}. For NM, we consider two initializations including ACMA and Root-MUSIC as in \cite{leshem2000maximum}. The root mean square error (RMSE) is defined as $\sqrt{\frac{1}{P}\sum^{P}_{p=1}\norm{\widehat{\m{\theta}}_p-\m{\theta}}^2_2}$ where $\widehat{\m{\theta}}_p$ is the vector of estimated DOAs in the $p$-th trial. The signal-to-noise ratio (SNR) is defined as $10\log_{10}\sbra{\frobn{\m{X}^\star}^2/\frobn{\m{E}}^2}$.  
	
	\emph{Experiment 2: Tightness of Relaxation.} To show tightness of the relaxation from the problem in \eqref{eq:p_S1} to the problem in \eqref{eq:p_S2}, it suffices to show that the number of recovered sources, given by $\rank(\cT \widehat{\m{t}})$, equals the true source number $K$. We consider a ULA composed of $N=15$ sensors with the number of snapshots $L=5$. Then $K=3$ sources are randomly generated with random phase and the signal is contaminated by i.i.d. complex Gaussian noise, where we consider SNR$\in \set{0,10,\ldots,40}\text{dB}$. Let $\sigma_K$ and $\widehat{\sigma}_K$ be the $K$-th largest singular value of the truth $\cT\m{t}$ and the solution $\cT\widehat{\m{t}}$ of the relaxed problem in \eqref{eq:p_S2}, respectively. We try $100$ Monte Carlo trials for each SNR and compute the ratio $\widehat{\sigma}_K / \sigma_K$ and $\widehat{\sigma}_K / \widehat{\sigma}_{K+1}$ for each trial. The histogram plots are presented in Fig.~\ref{fig:Histogram}. It is seen that $\widehat{\sigma}_K$ and $\sigma_K$ share approximately the same order of magnitude, which means $\rank(\cT\widehat{\m{t}})\ge K$, and $\widehat{\sigma}_K$ is much greater than $\widehat{\sigma}_{K+1}$ overall (within numerical precision), which implies $\rank(\cT\widehat{\m{t}})\le K$. This validates our analysis in Section \ref{sec:prop}.
	
	\begin{figure}
		\centerline{\includegraphics[width=11cm]{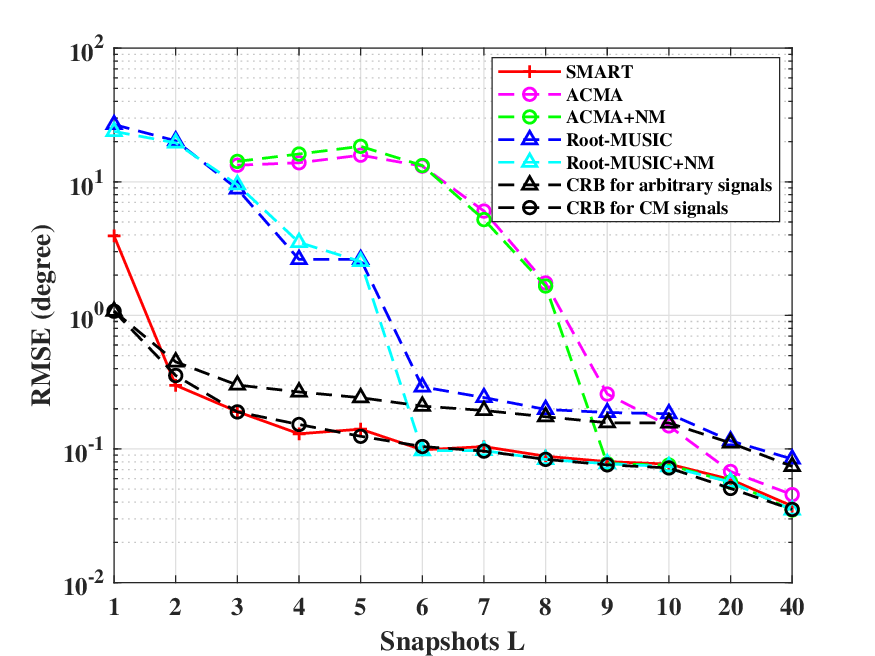}}
		\caption{RMSE in degree versus the number of snapshots $L$ for a ULA with $N=15$ elements, SNR $=20$dB and $\m{\theta} = [-2^{\circ},1^{\circ},30^{\circ}]^T$.}
		\label{fig:RMSEvsL}
	\end{figure}
	
	\emph{Experiment 3: RMSE vs. the Separation $\Delta_{\theta}$.} 
	For a ULA composed of $N=15$ sensors, we consider $K=2$ CM signals with random phases. The first source is fixed at $\theta_1 = -2^{\circ}$ and the second source is located at $\theta_2 = -2^{\circ}+\Delta_{\theta}$. We set the number of snapshots $L = 3$ and SNR $=20$dB. We say that the $K$ sources are successfully resolved if $\abs{\widehat{\theta}_k-\theta_k} < \min\lbra{ \abs{\theta_k - \theta_{k-1}}, \abs{\theta_{k+1} - \theta_{k}} } / 2$ for $k=1,\ldots,K$.
	A total of $P=200$ Monte Carlo trials are conducted and then averaged to produce each simulated point in the following figures.
	It is seen from Fig.~\ref{fig:RMSEvsResolution}(a) that SMART shows good convergence behaviour and the iteration number decreases as the increase of the separation of two DOAs. As shown in Fig.~\ref{fig:RMSEvsResolution}(b), the CRB for CM signals is lower than the CRB for arbitrary signals. It is shown in Fig.~\ref{fig:RMSEvsResolution}(b)-(c) that the performance of ACMA, Root-MUSIC and NMs initialized by them is poor. This is because ACMA and Root-MUSIC exploit either the Vandermonde structure or the CM structure only and NMs are highly sensitive to the initialization. The proposed SMART succeeds to resolve all sources in all trials and its RMSE attains or is slightly better than the CRB for CM signals (note that the maximum likelihood estimator is generally biased and can possibly be better than the CRB \cite{mardia1999bias}). 
	We also plot the curves of CPU computational time in Fig.~\ref{fig:RMSEvsResolution}(d). It can be seen that SMART shows comparable speed as ACMA and NM initialized by ACMA. 
	
	\emph{Experiment 4: RMSE vs. the Number of Snapshots $L$.} 
	We consider three CM signals with DOAs being $\m{\theta} = [-2^{\circ},1^{\circ},30^{\circ}]^T$. The SNR is $20$dB.
	It is seen from Fig.~\ref{fig:RMSEvsL} that Root-MUSIC can only reach the CRB for arbitrary signals since it does not use the CM structure. As a suboptimal two-step method, ACMA approaches the CRB for CM signals when the snapshots are sufficient, i.e., $L\ge 20$. As a parameter-domain MLE method, NMs initialized by ACMA and Root-MUSIC outperform ACMA and match the CRB for CM signals when $L\ge 9$ and $L\ge 6$, respectively. Thanks to the full exploitation of the Vandermonde and CM structures, SMART only requires $L\ge 2$ snapshots to attain the CRB for CM signals. 
	
	\begin{figure}
		\centerline{\includegraphics[width=11cm]{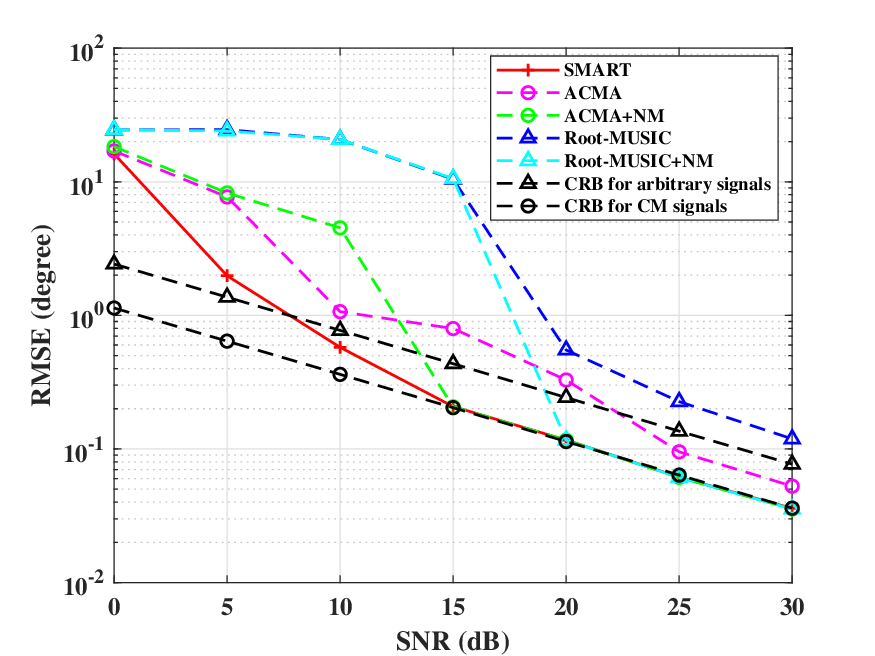}}
		\caption{RMSE in degree versus SNR for a ULA with $N=15$ elements, $L=3$ and $\m{\theta} = [-2^{\circ},1^{\circ}]^T$.}
		\label{fig:RMSEvsSNR}
	\end{figure}

	\emph{Experiment 5: RMSE vs. SNR.} 
	Fig.~\ref{fig:RMSEvsSNR} depicts the RMSE in degree versus SNR for a ULA with $N=15$ elements and $L=3$. We fix the DOAs $\m{\theta} = [-2^{\circ},1^{\circ}]^T$, the signal moduli $\m{b} = [1.79,2.62]^T$ and the signal phases $\m{\phi} = \begin{bmatrix} -2.4 & -1.52 & 1.90 \\ -1.45 & -1.58 & 1.22 \end{bmatrix}$. It is seen that ACMA performs better than Root-MUSIC and gets close to the CRB for CM signals under high SNR. The RMSEs of NMs initialized by Root-MUSIC and ACMA match the CRB for CM signals when SNR $\ge20$ dB and SNR $\ge15$ dB, respectively. SMART achieves the CRB for CM signals when SNR $\ge15$ dB and outperforms other methods. This verifies the statistical efficiency of the SMART estimator again. 
	
	\emph{Experiment 6: The SLA Case.} 
	We consider a SLA composed of $M=10$ sensors with $N=15$. In particular, we let the sensors index set $\Omega=\lbra{1,2,3,5,6,7,8,9,10,15}$ and $L=5$. We also fix the DOAs $\m{\theta} = [-2^{\circ},1^{\circ}]^T$, the signal moduli $\m{b} = [0.54,1.30]^T$ and the signal phases $\m{\phi} = \begin{bmatrix} -0.63 & 0.70 & 2.45 & 2.54 & 0.83 \\ 2.32 & -1.38 & 1.52 & 0.61 & -0.36 \end{bmatrix}$. Root-MUSIC and the associated NM are omitted since Root-MUSIC is not applicable in this case. 
	It is shown in Fig.~\ref{fig:RMSEvsSNR_M10} that as in the ULA case, SMART reaches the CRB for CM signals when SNR $\ge15$ dB and achieves the best performance. 
	
	\begin{figure}
		\centerline{\includegraphics[width=11cm]{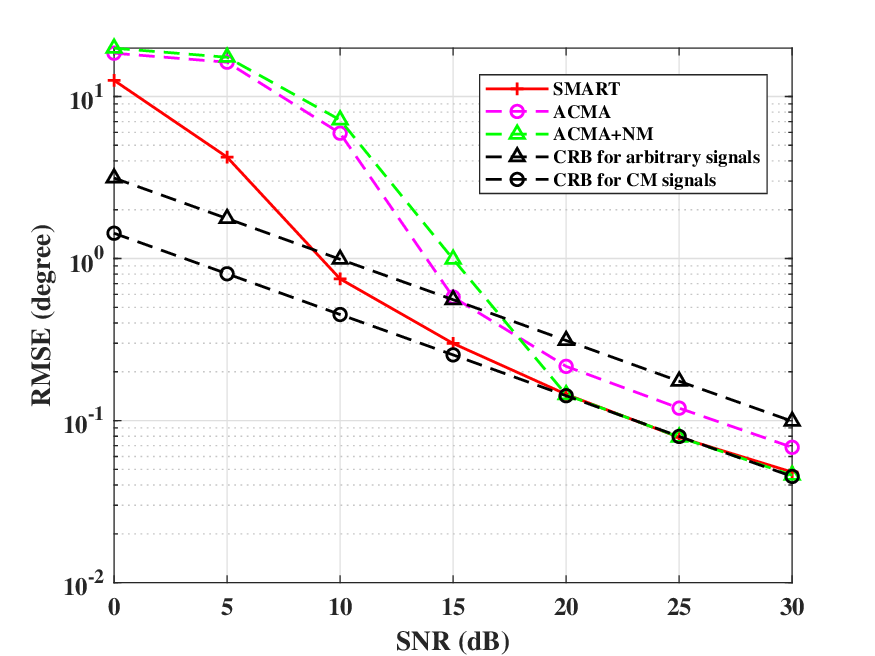}}
		\caption{RMSE in degree versus SNR for a SLA with $M=10$ elements, $L=5$ and $\m{\theta} = [-2^{\circ},1^{\circ}]^T$.}
		\label{fig:RMSEvsSNR_M10}
	\end{figure}
	
	\section{Conclusion} \label{sec:conclusion}
	In this paper, we proposed an SMART approach to make full use of the Vandermonde and CM structures for CM DOA estimation. By proposing a novel rank-constrained Hankel-Toeplitz characterization, the CM DOA estimation problem was formulated as a rank-constrained feasibility problem in the noiseless case and a rank-constrained minimization problem in the noisy case. Both of them were further generalized to the SLA case and were effectively solved by the ADMM algorithm. It is demonstrated by extensive numerical results that the proposed approaches have good performance in terms of source localization and statistical efficiency as compared to state-of-the-art algorithms. 
	
	The ADMM algorithm adopted in the proposed SMART needs to compute a truncated eigen-decomposition in each iteration, which limits the computational speed. It is of interest in further studies to develop faster nonconvex algorithms based on recent progress on low-rank matrix recovery \cite{cai2018exploiting}. Besides, since SMART requires a priori the number of sources, incorporating model order selection techniques \cite{stoica2004model} into SMART is also another future direction. Notice that under the assumption of uncorrelated sources, more sources than sensors can be localized using a SLA; see \cite{yang2022robust} and references therein. It is interesting to investigate whether the use of both CM and uncorrelatedness priors leads to even more locatable sources.
	
	\appendix
	\subsection{Proof of Proposition \ref{prop:1}} \label{append:1}
	\begin{proof}
		When $K=1$, it follows from \eqref{eq:range} that we must have $\rank\sbra{\cT\m{t}} = \rank\sbra{\cM(\m{X}_{:,l},\m{t})}=1$, and then $\cS'_{\text{HT}} = \cS_{\text{HT}} = \cS_0$. We next consider the case $K>1$. According to the relationship between $\rank\sbra{\cT\m{t}}$ and $\rank\sbra{\cM(\m{X}_{:,l},\m{t})}$, we divide $ \cS'_{\text{HT}}$ into two sets as follows:
		\equ{
			\cS'_{\text{HT}} = \cS_{\text{HT}} \cup \cS''_{\text{HT}}, \label{eq:S'HT}
		}
		where
		\equ{
			\begin{split}
				\cS''_{\text{HT}} & = \lbra{{\m{X}: \ \rank\sbra{\cT\m{t}} = K' < \rank\sbra{\cM(\m{X}_{:,l},\m{t})},  } \right. \\
					& \qquad \quad \ \left.  
					{ \cM(\m{X}_{:,l},\m{t}) \in \bS^+_K, l=1,\ldots,L,\;\m{t} \in \bC^{n}} }.
			\end{split}
		}
		
		For $\m{X} \in \cS''_{\text{HT}}$, it can be shown by arguments similar to those in the proof of Theorem \ref{thm:K<n} that there exist distinct $\lbra{\theta_k}^{K'}_{k=1}$, $\lbra{b_k > 0}^{K'}_{k=1}$ and $\lbra{S_{k,l}\in \bC}^{K',L}_{k=1,l=1}$ such that 
		\equ{
			\cT\m{t} = \m{A}\m{B}\m{A}^H \ \text{and} \ \cH\m{X}_{:,l}=\m{A}\bS^l\m{A}^T, l=1,\ldots,L, \label{eq:TtHx}
		}
		where $\m{A}$ here is an $n \times K'$ Vandermonde matrix, $\m{B}=\diag\sbra{\m{b}} \in \bR^{K' \times K'}$ with $\lbra{b_k > 0}$ and $\m{\bS}^l=\diag\sbra{\m{S}_{:,l}} \in \bC^{K' \times K'}$. Then for each $l=1,\ldots,L$, we have the decomposition \eqref{eq:decom_2}.
		Denote $\m{C}^l = \begin{bmatrix} \m{B} & \bS^l \\ \overline{\bS^l} & \m{B} \end{bmatrix}$. Using again the PSDness constraints in \eqref{eq:S_2}, we arrive at $ \m{C}^l \ge \m{0}$, and thus the Schur complement 
		\equ{
			\m{B} - \overline{\bS^l} \m{B}^{-1} \bS^l \ge \m{0},
		}
		or equivalently,
		\equ{
			b_k \ge \abs{S_{k,l}}, \ k=1,\ldots,K'. \label{eq:ge}
		}
		Further, since $\rank\sbra{\cM(\m{X}_{:,l},\m{t})} \le K$, we have 
		\equ{
			\rank \sbra{\m{C}^l} \le K. \label{eq:C}
		} 
		Given $l$, we denote two subsets $\mathscr{K}_1^l $ and $\mathscr{K}_2^l$ of the set $\lbra{1,\ldots,K'}$ as follows:
		\equ{
			\mathscr{K}_1^l  = \lbra{ k: \abs{S_{k,l}} < b_k }  \ \text{and} \
			\mathscr{K}_2^l  = \lbra{ k: \abs{S_{k,l}} = b_k }.
		}
		It follows from \eqref{eq:ge} and \eqref{eq:C} that
		\equ{
			\abs{\mathscr{K}_1^l}+\abs{\mathscr{K}_2^l}=K' \quad \text{and} \quad		2\abs{\mathscr{K}_1^l}+\abs{\mathscr{K}_2^l} \le K. \label{eq:num}
		}
		In other words, for each $k\in \mathscr{K}_1^l \sbra{\text{or} \; \mathscr{K}_2^l}$, the $k$-th and $(k+K)$-th columns of $\m{C}^l$ are linearly independent (or dependent), which contributes the rank of $\m{C}^l$ by two (or one). It follows from \eqref{eq:num} that 
		\equ{
			\abs{\mathscr{K}^l_1} \le \min\lbra{K', K - K'}. \label{eq:num2}
		}Using \eqref{eq:TtHx}, we have $\m{X}\in \cS''_0$ with $\Psi_{k,l} = S_{k,l} / b_k$ and then $\cS''_{\text{HT}} \subset \cS''_0$. 
		
		Conversely, for any $\m{X}\in \cS''_0$, we can easily verify that $\m{X}\in \cS''_{\text{HT}}$ by arguments similar to those in the proof of Theorem \ref{thm:K<n}. Therefore, we have 
		\equ{
			\cS''_{\text{HT}} = \cS''_0. \label{eq:''}
		}
		
		Combing \eqref{eq:S'HT}, \eqref{eq:''} and Theorem \ref{thm:K<n}, we conclude  
		\equ{
			\cS'_{\text{HT}} = \cS'_0,
		}
		completing the proof.
	\end{proof}
	
	\subsection{Derivations of \eqref{eq:t_HTs} and \eqref{eq:x_HT}} \label{append:2}
	
	\subsubsection{Derivation of \eqref{eq:t_HTs}} 
	According to the definition of $\cM^l=\cM(\m{X}_{:,l},\m{t})$ in \eqref{eq:HT2}, the problem regarding $\m{t}$ in \eqref{eq:2} is rewritten as
	\equ{
		\begin{split}
			\m{t}_{j+1} 
			& = \argmin_{\m{t}} \sum^L_{l=1} \frobn{\cM^l - \m{W}^l}^2 \\
			& = \argmin_{\m{t}} \sum^L_{l=1} \sbra{ \frobn{ \cT\overline{\m{t}} - \m{W}^l_1}^2 + \frobn{ \cT\m{t} - \m{W}^l_3}^2 }.
		\end{split} \label{eq:pt}
	}
	For an $n\times n$ Hermitian matrix $\m{C}$, it follows from the definition of the Hermitian Toeplitz operator $\cT$ in Section \ref{sec:org} that the optimal solution of $\min_{\m{t}} \frobn{\cT\m{t}-\m{C}}^2$ is given by $\m{t} = \m{D}_{\cT}^{-1}\cT^H\m{C}$. Consequently, equating the derivative of the objective function in \eqref{eq:pt} with respect to $\m{t}$ to zero yields the solution in \eqref{eq:t_HTs}. 
	
	\subsubsection{Derivation of \eqref{eq:x_HT}}
	By substituting \eqref{eq:loss} into the \eqref{eq:2}, the problem regarding $\m{X}$ is given by
	\equ{
		\begin{split}
			\m{X}_{j+1} 
			& = \argmin_{\m{X}} \frobn{\m{X}-\m{Y}}^2 \\
			& \quad + \frac{\rho}{2} \sum^L_{l=1} \sbra{ \frobn{\cH\m{X}_{:,l}-\m{W}^l_2}^2 + \frobn{\cH\overline{\m{X}_{:,l}}-\sbra{\m{W}^l_2}^H}^2 } \\
			& = \argmin_{\m{X}} \frobn{\m{X}-\m{Y}}^2 + \rho \sum^L_{l=1} \frobn{\cH\m{X}_{:,l}-\m{W}^l_2}^2, 
		\end{split} \label{eq:pX}
	}
	where the second equality follows from the symmetry of Hankel matrices $\cH\m{X}_{:,l}\in \bC^{n\times n}$. Since the problem in \eqref{eq:pX} is separable for each $l$, the problem of the $l$-th snapshot of $\m{X}$ is given by
	\equ{
		\begin{split}
			\sbra{\m{X}_{:,l}}_{j+1} = \argmin_{\m{X}_{:,l}} \frobn{\m{X}_{:,l}-\m{Y}_{:,l}}^2 + \rho \frobn{\cH\m{X}_{:,l}-\m{W}^l_2}^2.
		\end{split}
	}
	Applying the definition of Hankel operator in Section \ref{sec:org} and equating the derivative of the objective function with respect to $\m{X}_{:,l}$ to zero yield
	\equ{
		\begin{split}
			2\sbra{ \m{X}_{:,l}-\m{Y}_{:,l} } + 2\rho \sbra{ \m{D}_{\cH}\m{X}_{:,l} - \cH^H\m{W}^l_2 }  = \m{0},
		\end{split}
	}
	resulting in \eqref{eq:x_HT}.
	
	\bibliographystyle{IEEEtran}
	\bibliography{SMART.bib}
\end{document}